\documentclass[aps,pre,reprint,superscriptaddress,nofootinbib]{revtex4-2}

\usepackage{balance}

\usepackage{xcolor}       
\usepackage{hyperref}	  
\hypersetup{colorlinks=true,linkcolor=blue,citecolor=blue,breaklinks={true}} 
\usepackage[normalem]{ulem}

\usepackage{graphicx}
\graphicspath{{./FIG/}}
\DeclareGraphicsExtensions{.eps,.png,.jpg,.jpeg,.bmp,.gif,.pdf}

\usepackage{tabularx,booktabs}	
\usepackage{multirow}	

\usepackage{algorithm}

\usepackage{amsmath}      		
\usepackage{amssymb}	  		
\usepackage{bm}		      		
\usepackage{physics}	  		
\usepackage{mathtools}			
\usepackage{dsfont}		  		
\usepackage{soul}				
\usepackage{relsize}			

\newcommand{\R}{\mathbb{R}}		
\newcommand{\C}{\mathbb{C}}		
\newcommand{\transp}{\mathsf{T}}					

\usepackage{amsthm}
\newtheorem{defin}{Definition}
\newtheorem{theor}{Theorem}
\newtheorem{corol}{Corollary}
\newtheorem{lemma}{Lemma}
\theoremstyle{definition}
\newtheorem{assump}{Assumption}
\newtheorem{remark}{Remark}
\newtheorem{example}{Example}

\newcommand*{\QEDA}{\hfill\ensuremath{\triangle}}   %

\newcommand{\arthur}[1]{\textcolor{black}{#1}}

\usepackage{lipsum}
\usepackage{changepage}

\makeatletter
\newcommand*{\balancecolsandclearpage}{%
	\close@column@grid
	\cleardoublepage
	\twocolumngrid
}
\makeatother

\begin{document}

\title{Functional observability and target state estimation in large-scale networks}

\author{Arthur~N.~Montanari}
\email{arthur.montanari@uni.lu}
\affiliation{Department of Physics and Astronomy, Northwestern University, Evanston, IL 60208, USA}
\affiliation{Graduate Program in Electrical Engineering, Federal University of Minas Gerais, Belo Horizonte, MG 31270-901, Brazil}
\affiliation{Luxembourg Centre for Systems Biomedicine, University of Luxembourg, Belvaux L-4367, Luxembourg}

\author{Chao~Duan}
\email{chao.duan@northwestern.edu}
\affiliation{Department of Physics and Astronomy, Northwestern University, Evanston, IL 60208, USA}

\author{Luis~A.~Aguirre}
\affiliation{Departament of Electronic Engineering, Federal University of Minas Gerais, Belo Horizonte, MG 31270-901, Brazil}

\author{Adilson E. Motter}
\email{motter@northwestern.edu}
\affiliation{Department of Physics and Astronomy, Northwestern University, Evanston, IL 60208, USA}
\affiliation{Northwestern Institute on Complex Systems, Northwestern University, Evanston, IL 60208, USA}




\begin{abstract}
	The quantitative understanding and precise control of complex dynamical systems can only be achieved by observing their internal states via measurement and/or estimation.	In large-scale dynamical networks, it is often difficult or physically impossible to have enough sensor nodes to make the system fully observable. Even if the system is in principle observable,  high-dimensionality poses fundamental limits on the computational tractability and performance of a full-state observer. To overcome the curse of dimensionality, we instead require the system to be \textit{functionally observable}, meaning that a targeted subset of state variables can be reconstructed from the available measurements. Here, we develop a graph-based theory of functional observability, which leads to highly scalable algorithms to i) determine the minimal set of required sensors and ii) design the corresponding state observer of minimum order. Compared to the full-state observer, the proposed functional observer achieves the same estimation quality with substantially less sensing	and computational resources, making it suitable for large-scale networks. We apply the proposed methods to the detection of cyber-attacks in power grids from limited phase measurement data and the inference of the prevalence rate of infection during an epidemic under limited testing conditions.	The applications demonstrate that the functional observer can significantly 	scale up our ability to explore otherwise inaccessible dynamical processes on complex networks.
	\\
	\\
	\noindent
	DOI: \href{https://doi.org/10.1073/pnas.2113750119}{10.1073/pnas.2113750119}
\end{abstract}



\maketitle


Large-scale complex systems, including power grids, neuronal networks, and food webs, are often represented as sets of interconnected dynamical systems and referred to as \textit{dynamical networks}. 
Understanding the properties and control principles of dynamical networks allows the development of intervention strategies that can shape the behavior of these systems to achieve the desired functionality. As formalized by Wiener \cite{wiener2019cybernetics}, the fundamental mechanism enabling precise control of a dynamical system is \emph{feedback}, which involves sensors, signals, and actuators in a closed loop. A sensor provides immediate measurements of a particular variable of the system. 
As the dynamical network grows large, it becomes prohibitive to implement a sensor for each state variable, be it due to cost or physical constraints. 
For instance, our ability to measure each of the tens of billions of neurons present in a human brain is physically limited.
Likewise, infrastructure and operation costs may impede the placing of sensors in every node of a large technological system.
Therefore, the indirect estimation of the unmeasured states is essential for the control of large-scale dynamical networks.

\arthur{The property of a dynamical system that enables the reconstruction of the entire system state from its control inputs and sensor measurements is called \emph{observability} \cite{Kalman1959}. In particular, observability constitutes the necessary and sufficient condition that guarantees the \emph{existence} and enables the \emph{design} of full-state estimators---such as Luenberger observers \cite{Luenberger1966} and Kalman filters \cite{Kalman}.} 
Despite the success of state observers in uncountable engineering applications, high-dimensionality is still an obstacle to the direct use of these methods in large-scale dynamical networks \cite{Chen2014,Motter2015,Montanari2020}. This calls for different approaches and novel techniques \cite{Wang2002a,Liu2011,Cornelius2013,Liu2013c,Fiedler2013a,Zanudo2017,Aguirre2018} to overcome the lack of scalability of existing methods.
Based on a graph-theoretic approach to controllability \cite{Lin1974a}, Liu \textit{et al.} \cite{Liu2011} presented an efficient method that, by duality, can be used to determine a minimum set of sensor nodes required to guarantee the observability of complex networks. 
However, even if a minimum set of sensors is used, the state observer will have the same dimension as the entire network, making its design and implementation computationally expensive in large-scale systems. Moreover, a minimum set of sensor nodes does not guarantee good quality for the full-state reconstruction in higher-order systems \cite{Sun2013,Pasqualetti2013,Haber2017,Guan2018,Montanari2019}.

For many real-world problems, estimating the entire state vector of a high-dimensional system is not necessary or even desirable \cite{Motter2015}. It is often sufficient to focus on a particular subset of nodes of interest. For instance, in decentralized control strategies applied to network systems, each controller only requires feedback signals from a fraction of nodes in the neighborhood determined by the corresponding controlled area \cite{OlfatiSaber2004,Xue2018}. This is also true for the detection and monitoring of unforeseen failures and cyber-attacks, which finds several applications in supply networks \cite{Pasqualetti2013a}, power grids \cite{Zhang2013,Singh2014}, and autonomous vehicle coordination \cite{Vivek2019}. Similarly, in biomedical applications, estimation (diagnosis) and control (intervention) often require monitoring a reduced set of variables in the~respective networks \cite{Barabasi2011,Vinayagam2016}. Examples include regulatory network states associated with cancer \cite{Cornelius2013} and brain network states associated with Parkinson's disease \cite{Hammond2007} and epilepsy \cite{Lehnertz2009}.


These practical problems motivate the concept of functional observability \cite{Fernando2010,Jennings2011}, which characterizes the existence of a \emph{functional observer} capable of reconstructing a targeted subset of state variables from a limited number of sensors\textemdash even when the network is not completely observable.
\arthur{Functional observability can be related to the concept of \textit{target controllability} \cite{Gao2014,Klickstein2017,Czeizler2018,Li2018,Commault2019}, which establishes the conditions for the existence of a controller capable of steering a targeted subset of state variables and has been applied to problems of drug target identification \cite{Wu2015}. However, despite the duality between (complete) controllability and observability \cite{Kalman1959}, such duality does not hold between target controllability and functional observability. This is the case because, as we show, the state estimation of a subset of variables requires a stronger condition than the dual condition to the control of a subset of variables.
Even though the design of functional observers is a problem that dates back to the 1970s \cite{Luenberger1971,Darouach2000,Trinh2012}, previous studies on functional observability \cite{Fernando2010,Fernando2010b,Jennings2011} were based on numerical rank-based conditions, without explicitly taking advantage of the network topology, and thus do not lead to scalable algorithms applicable to large-scale networks.}

In this paper, we develop a graph-theoretic characterization of functional observability and the associated algorithms for sensor placement and observer design, making it possible to accurately estimate the target states of a large-scale dynamical network using minimal sensing and computational resources. The contributions of this work are threefold.
First, we propose the new concept of \textit{structural functional observability}, which can be seen as a generalization of Lin's structural observability \cite{Lin1974a}. This allows us to rigorously establish graph-theoretic conditions for functional observability equivalent to the original rank-based conditions \cite{Jennings2011}.
Second, based on the proposed theory, two highly scalable algorithms are developed to solve the sensor placement and observer design problems. The first algorithm determines a minimal set of sensors placed on a dynamical network to ensure  functional observability with respect to a given set of target nodes. After the sensor 
placement is decided, the second algorithm designs a minimum-order functional observer whose output converges asymptotically to the target states, achieving accurate estimation. 
Third, we demonstrate the advantages of the proposed methods with two concrete applications: the cyber-security of power grids and the monitoring of epidemic spreading. 
In power grids, we show that the proposed functional observers can be implemented as active detectors of cyber-attacks,
effectively providing state estimates that allow for cross-validation among different information sources and the detection of fake measurement data in real-time. 
In epidemics such as the COVID-19 pandemic, we demonstrate that the proposed functional observer can infer the fraction of infected population in areas where testing is limited from the data collected in areas with sufficient testing\textemdash moreover, our algorithms can also guide the optimal allocation of limited testing resources.

\section*{Results}

\noindent\textbf{Complete and functional observability of dynamical systems.}
\label{sec.functobsv}
A general linear dynamical system can be written as
\begin{equation}
\begin{aligned}
\begin{cases}
\dot{\bm x} = A\bm x + B\bm u, \\
\bm y = C\bm x,
\end{cases}
\end{aligned}
\label{eq.linearsys}
\end{equation}

\noindent
where $\bm x\in\R^n$ is the vector of state variables, $\bm u\in\R^p$ accounts for the control inputs or environmental influences, and $\bm y\in\R^q$ represents the direct measurement from available sensors in the system. 
Matrix $A$, which is referred to as the \textit{system matrix}, encompasses the nodal dynamics and network interactions\textemdash and can thus correspond to an adjacency matrix, a Laplacian matrix or, more generally, a Jacobian matrix of the system.
The system is completely observable if it is possible to reconstruct the entire state trajectory $\bm{x}(t)$ from the input vector $\bm{u}(t)$ and measurement vector $\bm{y}(t)$. Complete observability is guaranteed when the $nq\times n$ observability matrix
\begin{equation}
\mathcal O = \begin{bmatrix}
C^\transp & (CA)^\transp & (CA^2)^\transp & \ldots & (CA^{n-1})^\transp
\end{bmatrix}^\transp
\label{eq.obsvmatrix}
\end{equation}

\noindent
has full rank \cite{Chi-TsongChen1999}, i.e., $\rank(\mathcal O)=n$. Under this rank condition, there exist straightforward methods to design a full-state observer. Such observer is an auxiliary dynamical system 
whose states converge asymptotically to those of the original system \eqref{eq.linearsys} when taking $\bm{y}$ and $\bm{u}$ as inputs, providing an estimation of the state vector $\bm x$. Since the direct measurement $\bm y$ already contains $q$ linear combinations of state $\bm x$, only $(n-q)$ state variables are required to be reconstructed, which can be accomplished by a reduced-order state observer, which we refer to as the Luenberger observer \cite{Luenberger1966} (see Methods for details).

In practice, it is often unnecessary to estimate the entire state vector $\bm{x}$. Instead, only a lower-dimensional function $\bm z = F\bm x \in \R^r$ is of interest, where $r$ can be much smaller than $n$. Given the desirable ${F}$, functional observability characterizes the system property that enables the reconstruction of $\bm z$ from $\bm u$ and $\bm y$ \cite{Fernando2010}. The system is functionally observable if and only if \cite{Jennings2011,Rotella2016a}
\begin{equation}
\rank
\begin{bmatrix}
\mathcal O \\ F
\end{bmatrix}
=
\rank\left(\mathcal O\right),
\label{eq.linearfuncobsv}
\end{equation}

\noindent
that is, if and only if the row space ${\rm row}(F)$ is a subspace of the observable space ${\rm row}(\mathcal O)$. Clearly, complete observability is a special case of functional observability for $F=I$. However, condition \eqref{eq.linearfuncobsv} only guarantees the theoretical existence of a functional observer \cite{Jennings2011}. It does not readily lead to an algorithm to design a functional observer \cite{Fernando2010,Rotella2016a}, for which two additional conditions must be satisfied \cite{Darouach2000}:
\begin{align}
	&
	\rank
	\begin{bmatrix}
	C \\ CA \\ F_0 \\ F_0A
	\end{bmatrix}	
	=
	\rank
	\begin{bmatrix}
	C \\ CA \\ F_0
	\end{bmatrix}
	,
	\label{eq.darouachcond1}
\\
	&
	\rank
	\begin{bmatrix}
	\lambda F_0 - F_0A \\ CA \\ C
	\end{bmatrix}	
	=
	\rank
	\begin{bmatrix}
	CA \\ C \\ F_0
	\end{bmatrix},
	\label{eq.darouachcond2}
\end{align}

\noindent
where $F_0\in\R^{r_0\times n}$ and condition \eqref{eq.darouachcond2} must hold for every eigenvalue $\lambda$ of $A$. If a triple $(A,C,F)$ satisfies condition $\eqref{eq.linearfuncobsv}$, then there exists some matrix $F_0$ whose row space contains that of $F$ (i.e., ${\rm row}(F_0)\supseteq{\rm row}(F)$) that satisfies conditions \eqref{eq.darouachcond1}--\eqref{eq.darouachcond2} for the triple $(A,C,F_0)$ \cite{Fernando2010}. A functional observer of order $r_0\geq r$ can be designed systematically once such a matrix $F_0$ is determined (Methods). Finding a matrix $F_0$ with the minimum number of rows $r_0$ satisfying these conditions is then a crucial problem since the functional observer order is directly related to the computational costs of its design and real-time simulation (as demonstrated below).

\bigskip
\noindent\textbf{Structural functional observability.}
\label{sec.graphfuncobsv}
The rank-based conditions \eqref{eq.linearfuncobsv}--\eqref{eq.darouachcond2} are not numerically stable and computationally efficient for the design of functional observers for large-scale systems. Here, we adopt a graph-theoretic approach that explicitly leverages the network structure of the dynamical system. The system matrix $A$ can be structurally represented as a corresponding \textit{inference graph} $\mathcal G(A)$ whose nodes are the internal state variables $\mathcal X = \{x_1,\dots,x_n\}$. The links in $\mathcal G(A)$ capture the interaction pattern among state variables: there is a link from $x_j$ to $x_i$ on the graph $\mathcal G(A)$ if $A_{ij}$ is non-zero. A node $x_j$ on graph $\mathcal G(A)$ is a sensor node if $C_{ij}\neq 0$ for some $i$, and a node $x_k$ is a target node if $F_{ik}\neq 0$ for some $i$. The sets of all sensor and target nodes are denoted $\mathcal S$ and $\mathcal T$, respectively. We assume that each sensor or target is only related to one internal state variable, meaning that each row of $C$ or $F$ has only one non-zero entry. Throughout, the terms \textit{nodes} and \textit{links} are used exclusively in connection with inference graphs, and are not to be confused with the \textit{vertices} and \textit{edges} which are the corresponding terms used for the network systems. Fig.~\ref{fig.functobsvexample}\textit{A} illustrates the representation of the inference graph $\mathcal G(A)$ and the set of nodes estimated by a minimum-order functional observer for the indicated sensor and target nodes in a 10-dimensional system.

\begin{figure*}[t]
	\centering
	\includegraphics[width=0.6\textwidth]{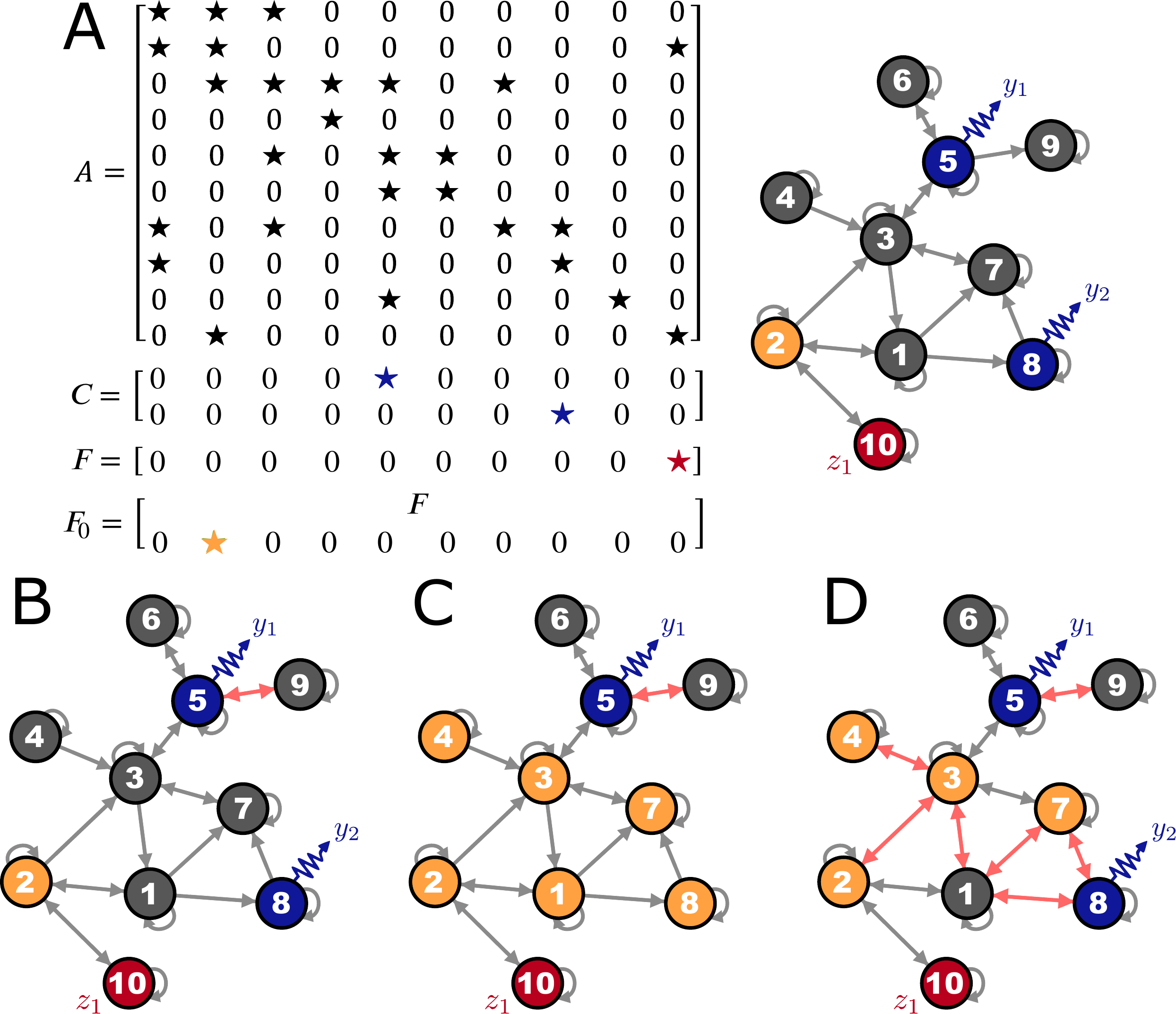}
	\caption{
		Structural functional observability of dynamical systems.
		(\textit{A})~System matrix $A$ and corresponding inference graph $\mathcal G(A)$. The set of state variables $\mathcal X =\{x_1,\ldots,x_{10}\}$ is represented by nodes on the graph, where the set of sensor nodes $\mathcal S = \{x_5, x_8\}$ (defined by $C$) is marked in blue and the set of target nodes $\mathcal T = \{x_{10}\}$ (defined by $F$) is marked in red.
		In this example, the system is unobservable ($\rank(\mathcal O)=9<n$), but it is  functionally observable ($\rank[\mathcal O^\transp \,\, F^\transp]^\transp=\rank(\mathcal O)$), enabling the design of a functional observer. 
		We show the minimum-order matrix $F_0$ such that $(A,C,F_0)$ satisfies conditions \eqref{eq.darouachcond1}--\eqref{eq.darouachcond2}. Non-zero elements in $F_0$ determine the minimum set of ``auxiliary'' nodes (highlighted in orange) whose states also have to be estimated by a functional observer in order to estimate the state of the target node. 
		The number of rows of $F_0$ is the order of the functional observer: $r_0=2$.
		(\textit{B})~Example of a system that is observable and hence functionally observable. Since the system is completely observable, a Luenberger observer can be designed. However, while such Luenberger observer has order $n-q=8$ because it estimates the state of every unmeasured node, a functional observer can estimate the target node $x_{10}$ with the much smaller order $r_0 = 2$.
		(\textit{C})~Minimum sensor set for functional observability. For the same graph and target set as in \textit{B}, the sensor node $x_5$ is a minimum sensor set required for the functional observability of $x_{10}$. However, the absence of sensor node $x_8$ increases the functional observer order to $r_0=7$.		
		(\textit{D})~Functional observability in a strongly connected graph. Due to the stronger connectivity, the functional observer order increases to $r_0 = 5$ (compared to \textit{B}), while remaining smaller than the Luenberger observer order $n-q=8$. Differences between graphs are indicated by the pink links.
	}
	\label{fig.functobsvexample}
\end{figure*}

A system given by the triple $(A,C,F)$ is said to be structurally functionally observable if there exists a functionally observable triple $(\tilde{A},\tilde{C},\tilde{F})$ which shares the \textit{same structure} as $(A,C,F)$. We define that triples $(A,C,F)$ and $(\tilde{A},\tilde{C},\tilde{F})$ have the same structure if, according to the representation described above, they share the same inference graph $\mathcal G(A)$, sensor set $\mathcal{S}$, and target set $\mathcal{T}$.
Thus, structural functional observability is purely determined by the state interaction structure encoded by graph $\mathcal G(A)$, the sensor node set $\mathcal{S}$, and the target node set $\mathcal{T}$\textemdash which are all independent of the specific numerical entries of $(A,C,F)$. In fact, if a triple $(A,C,F)$ is structurally functionally observable, a system that shares the same structure as $(A,C,F)$ is functionally observable with probability 1. 

This structural approach allows us to establish a graph-theoretic characterization of functional observability: 

\begin{adjustwidth}{0.5cm}{0.2cm}
\noindent\emph{A system $(A,C,F)$ is structurally functionally observable if and only if: 1) there exists a direct path from every target node to some sensor node, and 2) no target node is an element of a minimal subset of nodes with a dilation.}
\end{adjustwidth}

\noindent
A rigorous proof of this theorem is given in Section \ref{sec.structfuncobsv} of the Supporting Information. 
This result can be seen as a significant generalization of Lin's theory of structural controllability and, by duality, of structural observability \cite{Lin1974a}. Note that when $F=I$ (or, equivalently, $\mathcal{T}=\mathcal{X}$) our definition of structural functional observability reduces to Lin's structural observability. The latter states that a pair $(A,C)$ is structurally observable if and only if: 1') there exists a direct path from every state node to some sensor node, and 2') the corresponding graph has no dilations. For more background on structural observability, see Section \ref{sec.background} of the Supporting Information. We further illustrate this  characterization in Fig.\ \ref{fig.functobsvexample}\textit{A}, where the inference graph has no dilations due to the presence of self-links and thus satisfies conditions 2 and 2' for structural functional and complete observability, respectively. However, because node $x_9$ does not have a direct path to a sensor node, condition 1' for structural observability is not satisfied for the pair $(A,C)$. Nevertheless, condition 1 is satisfied for structural functional observability of the triple $(A,C,F)$ since one can identify a path from the target node $x_{10}$ to a sensor node ($x_{5}$ or $x_{8}$), hence rendering $(A,C,F)$  structurally functionally observable.

The above result lays a foundation for the functional observer design in large-scale dynamical networks. To enable the algorithm development, we further investigate two main design problems: 

\begin{enumerate}
\item How to select the minimum set of sensor nodes $\mathcal S$ such that a triple $(A,C,F)$ is structurally functionally observable? 

\item Given a structurally functionally observable triple $(A,C,F)$, how to determine the minimum-order matrix $F_0$ such that \eqref{eq.darouachcond1}--\eqref{eq.darouachcond2} are satisfied for $(A,C,F_0)$? (In other words, what is the minimum set of ``auxiliary'' state nodes that must be estimated along with the target nodes so that the systematic functional observer design is possible?)
\end{enumerate}

\noindent
The examples in Fig.~\ref{fig.functobsvexample}\,\textit{B}--\textit{D} illustrate that both questions are intertwined and inherently related to the structure of $\mathcal G(A)$.

In the following sections, we assume that no target node is an element of a minimal subset of nodes with a dilation. A sufficient condition for this latter assumption is that every target node has a self-link. The importance of including self-links in dynamical networks models, especially for state control and estimation applications, has been thoroughly discussed in the literature \cite{Cowan2012,Leitold2017,Montanari2020}. Indeed, dilations are not found in a broad range of dynamical networks, including  diffusively coupled systems. Our assumptions can be satisfied for applications in networks of coupled oscillators \cite{Arenas2008,Rodrigues2016,Eroglu2017}, power grids \cite{Dorfler2013,Nishikawa2015}, neuronal systems \cite{Izhikevich2004,Aguirre2017}, combustion networks \cite{Perini2012,Haber2017}, regulatory networks \cite{Mirsky2009,Mochizuki2013}, consensus problems \cite{OlfatiSaber2004}, and multi-group epidemiological models \cite{Colizza2006}.

\bigskip
\noindent\textbf{Minimum sensor placement for functional observability.}
\label{sec.msp}
According to the theory and assumptions discussed above, the minimal sensor placement problem is to determine a minimum set $\mathcal S$ such that there is a direct path in $\mathcal G(A)$ from every target node to some sensor node. We show that the minimum sensor placement problem can be formulated as a \textit{set cover problem}. For each candidate sensor node, let $\mathcal R_i$ denote the set of target nodes 
that have a direct path to the state node $x_i\in\mathcal X$. By this definition, the minimal sensor placement amounts to identifying the minimal sensor set $\mathcal{S}$ such that the union of the sets $\mathcal{R}_i$ for all $x_i\in \mathcal{S}$ covers the target set $\mathcal{T}$, i.e., $\cup_{x_i\in \mathcal{S}} \mathcal R_i \supseteq \mathcal{T}$. This is an NP-hard problem \cite{Corman1990}, to which we provide an approximate but highly scalable solution via Algorithm~\ref{alg.MSP} (Methods), where a breadth-first search determines $\mathcal R_i$ for each node $x_i \in \mathcal{X}$ and a greedy algorithm solves the set cover problem. Owing to the submodularity of the problem \cite{Fujito1999}, this approximation is guaranteed to be near-optimal.

\begin{figure}
	\centering
	\includegraphics[width=\columnwidth]{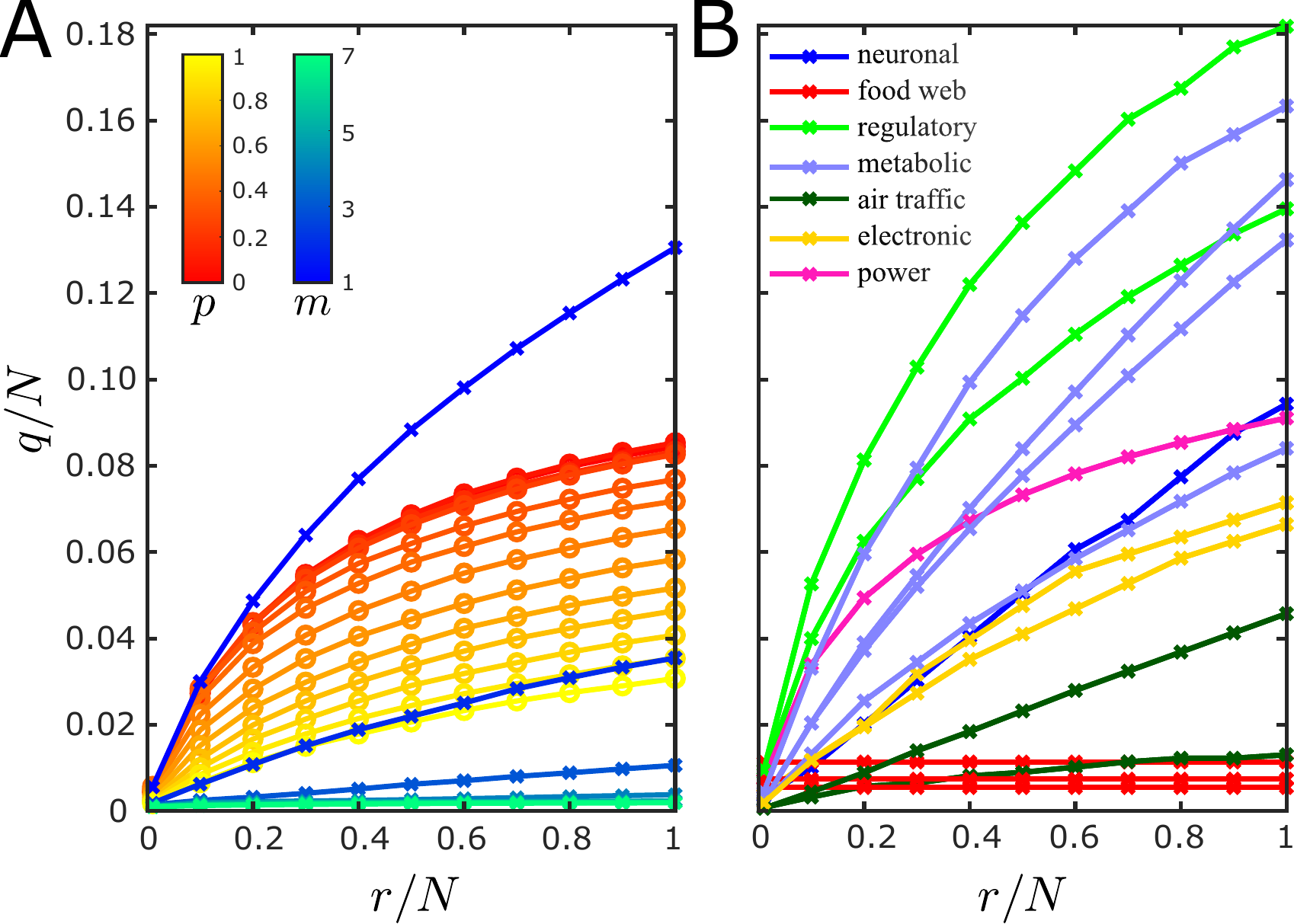}
	\caption{
		Minimum sensor placement in large-scale networks.
		Minimum number of sensors $q$ as a function of the number of target nodes $r$ (normalized by the network size $N$) in: (\textit{A})~randomly generated directed SW and SF networks and (\textit{B})~real-world networks. Each data point is an average over 100 realizations of randomly selected target nodes. The minimum set of sensor nodes $\mathcal S$ is determined using Algorithm~\ref{alg.MSP}. 
		The SW and SF networks were generated with $N=10^4$ vertices, where each vertex is a 3-dimensional subsystem (i.e., $n=3N$), while vertices in the real-world networks are assumed to be 1-dimensional subsystems (i.e., $n=N$). The SW networks were generated using the Newman-Watts model, where $p$ is the probability of adding a new edge, and the SF networks were generated using the Barabási-Albert model, where $m$ is the number of existing vertices a new vertex is connected to. See Methods for more details on the network datasets and models.}
	\label{fig.msp}
\end{figure}

Fig.~\ref{fig.msp} illustrates the application of Algorithm \ref{alg.MSP} to randomly generated small-world (SW) and scale-free (SF) networks as well as a selection of real-world networks. As expected, the results show that a smaller number of target nodes tends to require a smaller number of sensor nodes to guarantee the functional observability of a system. 
Fig.~\ref{fig.msp}\textit{A} shows, however, that the relation between the minimum set of sensor nodes and the number of target nodes depends on the network structure, where the number of sensor nodes can be substantially smaller if the network connectivity is larger (i.e., higher parameters $p$ and $m$ in SW and SF networks, respectively).  
Similar conclusions are also noted for the real-world networks shown in Fig.~\ref{fig.msp}\textit{B}. As $r$ approaches $n$, the minimum number of sensor nodes tends to the number required for complete observability. For example, in the metabolic networks analyzed, complete observability requires monitoring $8$--$16\%$ of all metabolites, which is consistent with previous findings \cite{Liu2013c}.
However, complete observability is often unnecessary for many biomedical applications since the number of biomarkers (e.g., target nodes whose activity is altered by a disease) is usually much smaller than the network size ($r\ll N$) \cite{Barabasi2011}. If only $1\%$ of the metabolites are biomarkers, then functional observability can be guaranteed by placing sensors in only $0.17$--$0.24\%$ of the state nodes.
Moreover, Fig.~\ref{fig.si.msp} (Supporting Information) shows that, for the same metabolic networks, around $78$--$85\%$ of all metabolites are observable from a single optimally placed sensor node. This means, in particular, that in applications where all target nodes belong to this set of observable nodes, functional observability can be achieved with even fewer sensor nodes than shown in Fig.~\ref{fig.msp}\textit{B} (where the targets were randomly chosen). In addition to their significance for biological and ecological networks, these results are also relevant for cyber-physical systems in engineering applications (e.g., power grids and transportation networks), where the monitoring and detection of potential failures or cyber-attacks are often required to be conducted in specific nodes.

\bigskip
\noindent\textbf{Minimum-order functional observer design.}
\label{sec.observerdesign}
After the sensor nodes have been selected, we need to further choose a matrix $F_0$ to enable the design of a functional observer. The theoretical problem of finding a minimum-order $F_0$ that satisfies conditions \eqref{eq.darouachcond1}--\eqref{eq.darouachcond2} was solved in the past decade \cite{Fernando2010}. However, a direct numerical implementation of the method \cite{Fernando2010b} is not scalable for high-dimensional systems because it iteratively uses singular value decomposition (SVD) to construct a matrix $F_0$ that satisfies condition \eqref{eq.darouachcond1} and is followed by a combinatorial search to augment the number of rows of $F_0$ in order to satisfy condition \eqref{eq.darouachcond2} (Supporting Information, Section \ref{sec.graphalgorithm}).
We circumvent these issues by adopting the structural approach described in the previous sections, in which we convert the rank-based conditions \eqref{eq.darouachcond1}--\eqref{eq.darouachcond2} 
into equivalent graph-theoretic ones. This is achieved by first noting that, if the corresponding graph of a dynamical system has a self-link in every target node (as assumed throughout), then condition \eqref{eq.darouachcond1} implies \eqref{eq.darouachcond2} with probability 1 for triples $(A,C,F_0)$ sharing the given structure
(Supporting Information, Corollary~\ref{corol.structdarouachcond}). In light of this, only condition \eqref{eq.darouachcond1} needs to be considered to determine $F_0$ and hence the combinatorial search is no longer needed. We thus propose Algorithm \ref{alg.findF0}  (Methods) as a highly scalable solution to determine matrix $F_0$ with the \textit{smallest} order possible by adding suitable rows to $F$ in such a way that \eqref{eq.darouachcond1} is satisfied. In Algorithm \ref{alg.findF0}, instead of invoking SVD, the rank condition \eqref{eq.darouachcond1} is verified by computing the maximum matching set of an associated bipartite graph. The algorithm is shown to have a computational complexity of order $O(n^{2.5})$ (Methods), which is a substantial improvement compared to the complexity of order $O(n^4)$ of the numerical procedure in \cite{Fernando2010b}. 

\begin{figure*}[t!]
	\centering
	\includegraphics[width=0.65\textwidth]{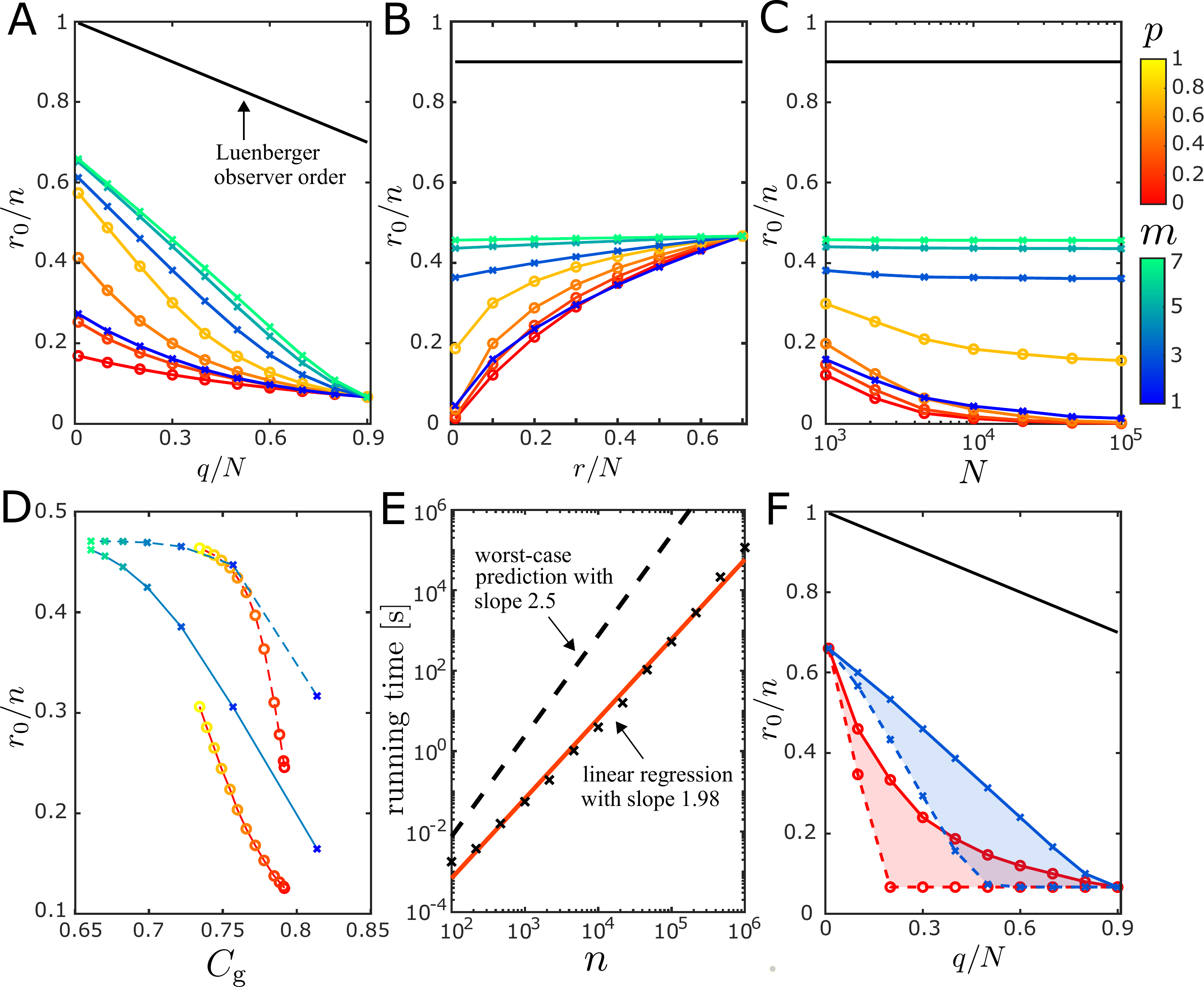}  
	\caption{
		Minimum-order functional observer design in large-scale networks.
		(\textit{A}--\textit{C})~Minimum functional observer order $r_0$ (normalized by the system dimension $n=3N$) as a function of the normalized number of sensor nodes $q/N$~(\textit{A}), normalized number of target nodes $r/N$~(\textit{B}), and network size $N$~(\textit{C}).
		The results are shown for random placement of sensor and target nodes in directed SW and SF networks (color coded by the respective parameters $p$ and $m$).
		The other parameters are set to $(N,r)=(10^4,0.1N)$ for (\textit{A}), $(N,q)=(10^4,0.3N)$ for (\textit{B}), and $(q,r)=(0.3N,100)$ for (\textit{C}).
		The black line indicates the Luenberger observer order $(n-q)$ for comparison. 
		(\textit{D})~Normalized order $r_0/n$ in directed (solid line) and undirected (dashed line) SW and SF networks as a function of the generalized clustering of the corresponding undirected graph (Methods), color coded by $p$ and $m$, for $(N,q,r)=(10^4,0.3N,0.1N)$.
		(\textit{E})~Running time of Algorithm \ref{alg.findF0} as a function of $n$ in directed SW networks for $(q,r,p)=(0.3N,0.1N,0.2)$.
		The simulations were implemented in MATLAB 
		and each network realization was run on a single core of an Intel Xeon CPU E7-8867 v4 at 2.4GHz.
		%
		(\textit{F})~Normalized order $r_0/n$ as a function of $q/N$ in undirected SW (red, $p=0$) and SF (blue, $m=3$) networks for randomly (solid line) and optimally (dashed line) placed sensors, where  $(N,r)=(100,0.1N)$.
		In all panels, each data point corresponds to an average over 100 independent realizations of the network, target placement, and sensor placement (except for the optimal placement in \textit{F}).
	}
	\label{fig.obsvorder}
\end{figure*}

Fig.~\ref{fig.obsvorder}\textit{A,B} illustrates the minimum order of the functional observer for randomly generated networks, determined by Algorithm \ref{alg.findF0}, as a function of the numbers of sensor and target nodes. 
On average, a larger sensor set $\mathcal S$ leads to a lower order $r_0$, whereas a larger target set $\mathcal T$ results in a higher order $r_0$. 
We note that, overall, functional observers are of much lower order compared to the corresponding Luenberger observers. This leads to a significant improvement in computation efficiency and scalability when designing and implementing observers in large-scale networks.
For a fixed number of target nodes, Fig.~\ref{fig.obsvorder}\textit{C} shows that the functional observer order normalized by the system dimension, ${r_0}/{n}$, tends to decrease as the network size increases  (in SF networks, 
$r_0/n$ exhibits weaker dependency on $N$ as $m$ increases). This means that the order reduction gained by the functional observer compared to the Luenberger observer increases with the network size. The extent of this gain depends, however, on other system properties, including the structure of the inference graph $\mathcal G(A)$, the choice of target nodes in $\mathcal T$, and how sensor nodes in $\mathcal S$ are placed. 
In particular, directed links, self-links, and clustering in $\mathcal G(A)$ tend to lead to a larger order reduction in the functional observer design for the random model networks considered. This is illustrated in Fig.~\ref{fig.obsvorder}\textit{D}, where it is shown for both directed and undirected networks that the functional observer order decreases sharply as a function of the generalized clustering $C_{\rm g}$, which is defined to account for both clustering and self-links (Methods). 
Interestingly, although directed graphs require a larger minimum set of sensor nodes to guarantee the structural functional observability of a system compared to undirected graphs (which only require only one sensor node \cite{Cowan2012}), directed graphs allow the design of functional observers of smaller orders. 
This result also highlights that Algorithm \ref{alg.findF0} brings computational improvement for both directed and undirected network applications compared to existing ones. Furthermore, Fig.~\ref{fig.obsvorder}\textit{E} illustrates how the running time of Algorithm~\ref{alg.findF0} scales with the network size, showing that it does not surpass our worst-case prediction.

The results shown in Fig.~\ref{fig.obsvorder}\textit{A--E} concern sensors and targets randomly placed in the inference graph. The sensor placement, in particular, was implemented by first finding the minimum set of sensor nodes for functional observability and then increasing $q$ with randomly placed sensor nodes.
As shown in Fig.~\ref{fig.obsvorder}\textit{F} for undirected networks, the functional observer order $r_0$ decreases on average as the number of sensors increases, even if the placement is random.
However, $r_0$ can be further reduced by optimizing the placement of the additional sensors. 
This is a computationally demanding bi-level optimization problem, which\textemdash for illustration purposes only\textemdash we solve using a (non-scalable) greedy algorithm (Methods). Fig.~\ref{fig.obsvorder}\textit{F} shows that such optimization indeed leads to a functional observer with a consistently smaller order compared to the average order for randomly placed sensors. Even though this specific result is illustrated in a lower-dimensional setting, we extrapolate from Fig.~\ref{fig.obsvorder}\textit{C} that this optimal sensor placement can be relevant for systems of any dimension.

\bigskip
\noindent\textbf{Comparative analysis of the observers.}
\label{sec.compobsv}
Fig.~\ref{fig.obsvestimation} compares the performances of the functional observer and Luenberger observer when estimating the target variables of a large-scale network. For the target state evolution illustrated in Fig.~\ref{fig.obsvestimation}\textit{A}, which is representative of a trajectory starting away from equilibrium, the transients of the target state estimation error $\|\bm z(t) - \hat{\bm z}'(t)\|$ are presented in Fig.~\ref{fig.obsvestimation}\textit{B} for both observers initialized with unknown initial conditions (Supporting Information, Section~\ref{sec.graphalgorithm}). It can be seen that the functional and Luenberger observers have similar dynamical behavior and that their estimates converge to the target states of the system. Statistical analysis further reveals that the two observers perform asymptotically close to each other even under the effects of modeling errors in the system matrix $A$ (Fig.~\ref{fig.obsvestimation}\textit{C}). (Note that if the system were functionally, but not completely, observable, then the target estimation error is only guaranteed to converge for the functional observer.)

\begin{figure*}[t!]
	\centering
	\includegraphics[width=0.55\textwidth]{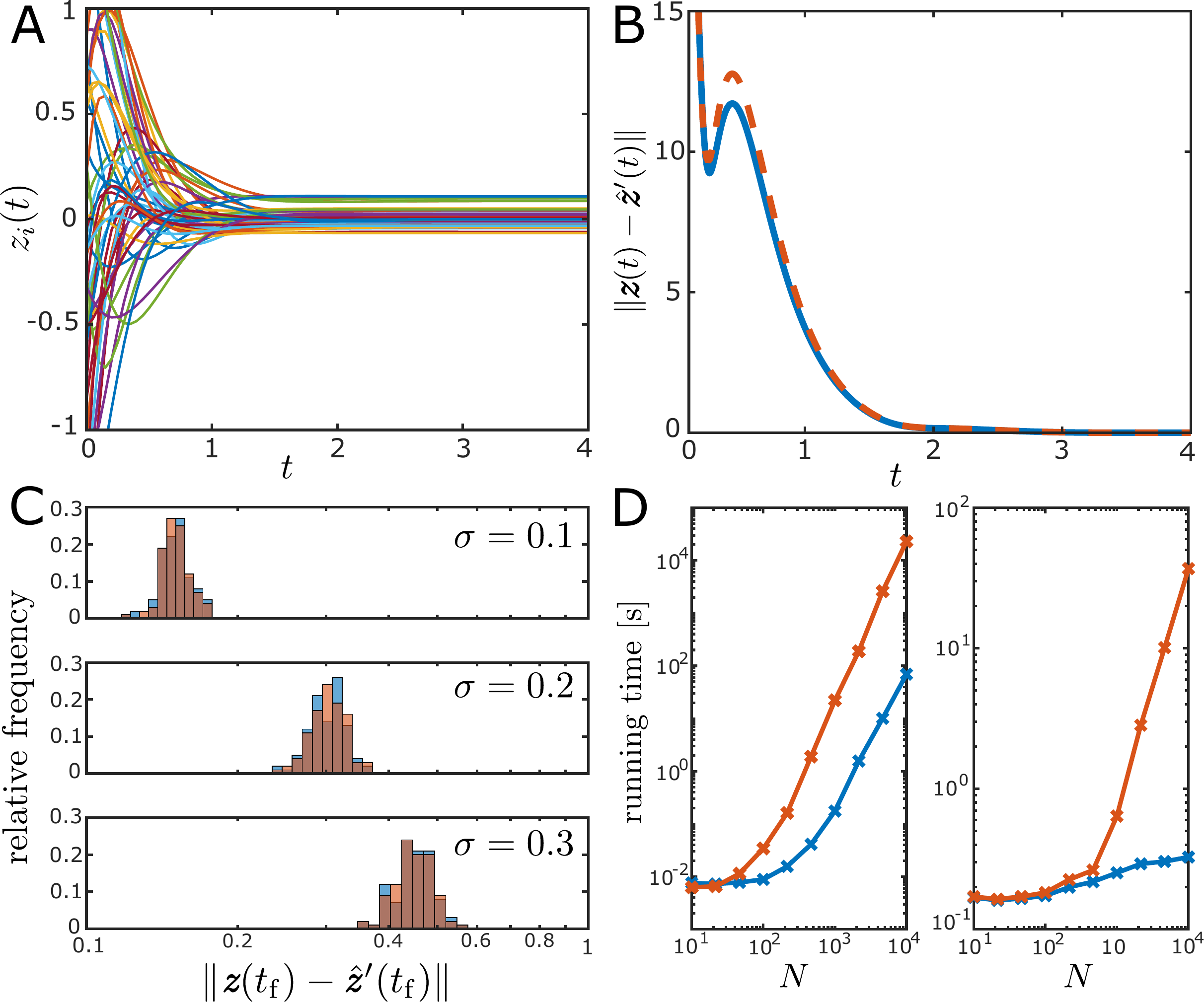}
	\caption{
		Comparative performance of the observers for target state estimation in large-scale networks.
		(\textit{A})~Dynamical evolution of target vector $\bm z(t) = F\bm x(t)$, where each color represents a different target variable $z_i(t)$.
		(\textit{B})~Dynamical evolution of the target state estimation error $\|\bm z(t) - \hat{\bm z}'(t)\|$, where $\bm z(t)$ is the ``true vector value'' of the target state and $\hat{\bm z}'(t)$ is the estimated target state provided by the functional (blue) or Luenberger (orange) observer initialized with random estimates.
		(\textit{C})~Histogram of the steady-state estimation error $\|\bm z(t_{\rm f}) - \hat{\bm z}'(t_{\rm f})\|$ for $t_{\rm f}=4$~s, where both observers are designed using a system matrix $\tilde A$ for different modeling errors $\sigma$. Here, each matrix entry is drawn from a uniform distribution as $\tilde A_{ij}\sim\mathcal U[(1-\frac{\sigma}{2}) A_{ij},(1+\frac{\sigma}{2}) A_{ij}]$.
		(\textit{D})~Running time of the design algorithms (left) and simulation time of the observer dynamics (right) as a function of the network size $N$ (computer specifications in Fig.~\ref{fig.obsvorder}\textit{E}).
		The color code in (\textit{C},\textit{D}) is the same as in (\textit{B}).
		In all panels, the results are shown for directed SW networks with randomly selected sensor and target nodes. The undeclared parameters are set as $(N,n,p,q,r) = (10^3,3N,0.2,0.3N,0.1N)$, where each parameter choice in (\textit{C,D}) corresponds to 100 realizations of the network, sensor and target selection, and modeling errors ($\sigma = 0$ in \textit{D}). See the Supporting Information, Section~\ref{sec.graphalgorithm}, for details on the simulations.
	}
	\label{fig.obsvestimation}
\end{figure*}

Overall, the numerical results show that the considerable order reduction demonstrated here for the functional observer design does not compromise its efficacy. Fig.~\ref{fig.obsvestimation}\textit{D}, on the other hand, shows that such order reduction significantly reduces the computational costs both in the design (Algorithm~\ref{alg.findF0}) and in real-time simulations of the functional observer. This computational advantage of functional observers makes them superior or even indispensable to observe large-scale networks, especially when continual re-design of the observer is expected due to the evolution of the system's equilibrium and/or network structure.


\bigskip
\noindent\textbf{Cyber-attack detection in power grids.}
\label{sec.app.powergrid}
The control of man-made technological systems, such as power grids, supply networks, interconnected autonomous vehicles, and swarms of robots, is supported by sensing and communication infrastructure. 
Decentralized control strategies \cite{Bakule2008}, such as wide-area 
control in power grids \cite{Xue2018}, are important to maintain system stability and, in particular, mitigate the impact of perturbations that could lead to large-scale failures \cite{Yang2017a}. However, such control strategies rely on  resilient communication networks between spatially distributed components, which are arguably more vulnerable to potential failures and cyber-attacks than the physical systems themselves. Indeed, there have been growing threats to cyber-security, including cyber-attacks to supervisory control and data acquisition (SCADA) systems, which led to the massive 2015 power outages in Ukraine \cite{Lee2016}, the 2000 Maroochy Water Services breach in Australia \cite{Slay2007}, the 2010 Stuxnet computer worm attack on Iran's nuclear program  \cite{Farwell2011}, and communication outages in the Western U.S. power grid in March 2019 \cite{NAERC2019}. 

Two common types of cyber-attacks are denial-of-service attacks (e.g., via signal jamming) and deception attacks (e.g., via data corruption) \cite{Amin2009}. Depending on the specifics of the attack, modeling of the physical system dynamics and transmitted data can still be used to design observers capable of recovering lost data through state estimation. Crucially, state estimation can also be used to detect deception attacks \cite{Teixeira2010,Pasqualetti2013a,Giraldo2018}, which is significant because such attacks are designed to evade detection.
We now show, in the context of power grids, how functional observers can be implemented for cyber-attack detection and data recovery, and the extent of their computational efficiency improvement compared to the traditional full-state estimators.

\begin{figure*}[t]
	\centering
	\includegraphics[width=\textwidth]{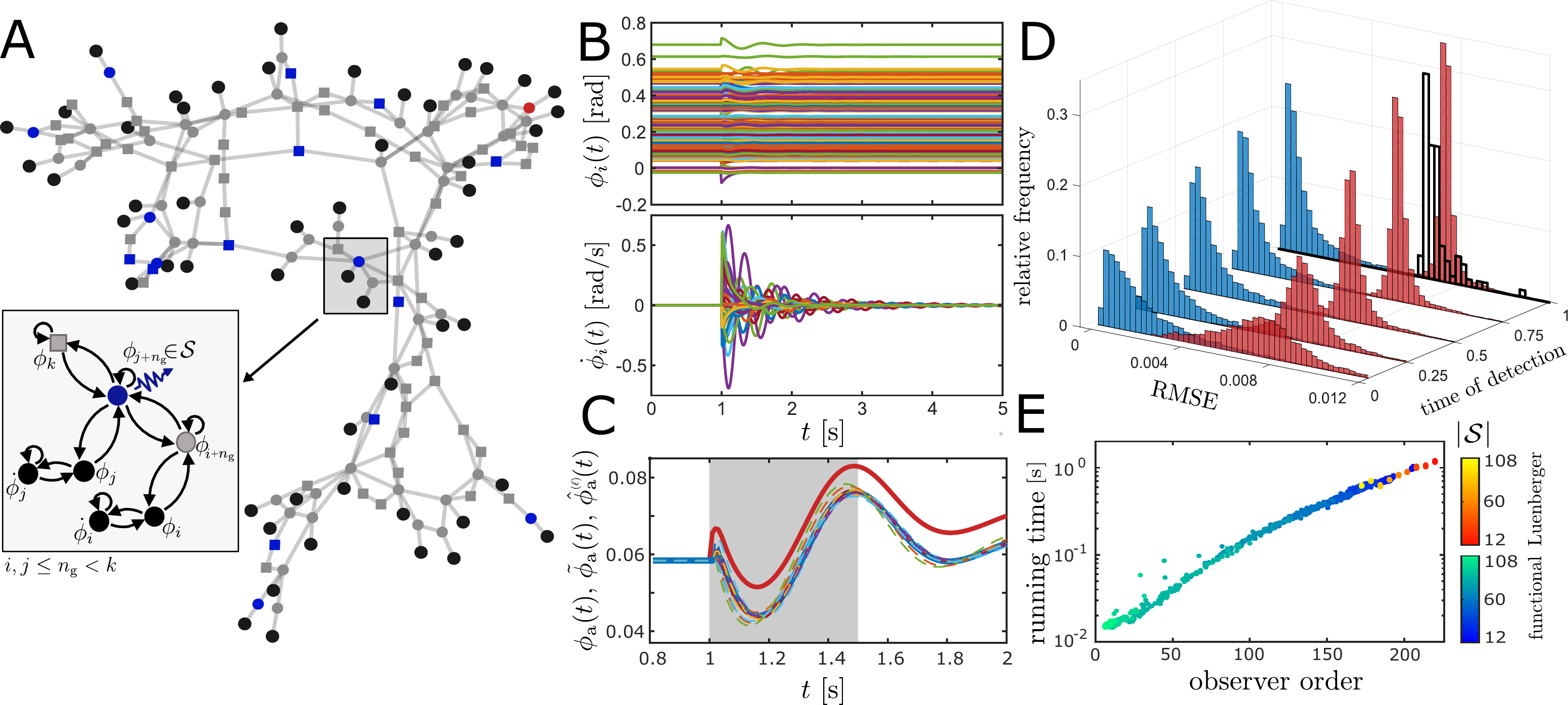}
	\caption{
		Deception-attack detection in a power grid.	
		(\textit{A})~Diagram of the IEEE-118 network, which comprises 118 buses with $n_{\rm g}=54$ generators (black circles), each connected to a terminal bus (gray circles), and $n_{\rm l} = 64$ loads (gray squares). PMUs are randomly placed in $|\mathcal S|=n_{\rm g}/3$ load and generator terminal buses (blue symbols). The deception attack tampers with the transmitted measurement $\bar{\phi}_{\rm a}(t)$ from the highlighted PMU (red circle).
		\arthur{The inset illustrates the inference graph of the dynamical system 
		\eqref{eq.generator}--\eqref{eq.load}, where each generator bus (vertex) is represented by two state nodes (phase $\phi_i$ and frequency $\dot\phi_i$ in \eqref{eq.generator}) and each generator terminal and load bus is represented by a single state node (phase $\phi_i$ in \eqref{eq.load}). Links with state nodes outside the highlighted neighborhood are omitted.}
		(\textit{B})~Dynamics of the oscillators' phases $\{{\phi}_i(t)\}_{i=1}^N$ and the generators' frequencies $\{\dot{\phi}_i(t)\}_{i=1}^{n_{\rm g}}$ over time $t$. An additive perturbation, drawn from the Gaussian distribution $\mathcal N(0,0.01)$, is applied to each generator's phase in steady state at $t=1$~s.
		(\textit{C})~Functional observer-based detection of a deception attack, where the transmitted measurement $\bar{\phi}_{\rm a}(t)$ of the terminal's phase $\phi_{\rm a}(t)$ (blue solid line) is replaced by the false data $\tilde{\phi}_{\rm a}(t)$ (red solid line). The data are reconstructed with the state estimates $\hat{\phi}_{\rm a}^{(\ell)}(t)$ provided by different functional observers (colored dashed lines), where the shaded window illustrates a time of detection $t_{\rm d}=0.5$ s.
		(\textit{D})~Histogram of the RMSE between the transmitted measurement $\bar{\phi}_{\rm a}(t)$ and the state estimates $\hat\phi_{\rm a}^{(\ell)}(t)$ provided by the functional observers $\ell=1,\ldots,100$ as a function of  $t_{\rm d}$. 
		The histograms are shown for the system under attack  
		(red) and not under attack 	
		(blue), where an estimate of the statistical properties of the latter can be inferred from training data.
		Each histogram comprises 10{,}000 data points corresponding to 100 independent realizations of the system perturbations with 100 functional observers designed for each realization. 
		The contoured histogram at $t_{\rm d} = 1$ s represents the hardest-to-detect attack (smallest median value) for the simulated perturbations.
		(\textit{E})~Running time of the observer design as a function of the observer order. The results are color coded for functional (blue scale) and Luenberger (red scale) observers as the number of sensors $|\mathcal S|$ is varied for 100 independent realizations, where the placement of the attack and PMUs is random in each realization and each data point is an average over 100 observers.
		In all simulations, the observers were designed to estimate the target node $\phi_{\rm a}$ using the system model \eqref{eq.generator}--\eqref{eq.load} linearized around the equilibrium point 
		and $\hat{\phi}_{\rm a}^{(\ell)}(t)={\phi}_{\rm a}(t)$ for $t<1$ s (see Methods for details).
	}
	\label{fig.powergrid}
\end{figure*}

The power-grid dynamics can be modeled as a structure-preserving network of coupled first- and second-order Kuramoto oscillators \cite{Dorfler2013,Nishikawa2015}. In this model, the generators dynamics are governed by the so-called swing equation,
\begin{equation}
\frac{2H_i}{\omega_{\rm R}}\ddot{\phi}_i + \frac{D_i}{\omega_{\rm R}}\dot{\phi}_i = P_i + \sum_{j=1,j\neq i}^{N}K_{ij}\sin(\phi_j-\phi_i),
\label{eq.generator}
\end{equation}
for $i=1,\ldots,n_{\rm g}$, and the dynamics of load buses and generator terminals are described as first-order phase oscillators,
\begin{equation}
\frac{D_i}{\omega_{\rm R}}\dot{\phi}_i = P_i + \sum_{j=1,j\neq i}^{N}K_{ij}\sin(\phi_j-\phi_i),
\label{eq.load}
\end{equation}
for $i=n_{\rm g}+1,\ldots,N$, where $n_{\rm g}$ is the number of generators, $n_{\rm l}$ is the number of load buses, $N = 2n_{\rm g}+n_{\rm l}$ is the number of oscillators (vertices), and $n=N+n_{\rm g}$ is the system dimension. Here, $\phi_i(t)$ is the phase angle of oscillator $i$ at time $t$ relative to the frame rotating at reference frequency $\omega_{\rm R}$, and $H_i$ and $D_i$ are the inertia and damping constants, respectively. In addition, $K_{ij}=V_iV_jB_{ij}$, where $B_{ij}$ is the susceptance of the transmission line connecting buses $i$ and $j$, and $V_i$ and $V_j$ are the voltage magnitudes at these buses. If there is no line connecting buses $i$ and $j$, $K_{ij}=0$. The power injection $P_i$ represents power generation for $P_i>0$ and power consumption for $P_i<0$.

We illustrate our framework on the IEEE-118 benchmark system given by the diagram in Fig.~\ref{fig.powergrid}\textit{A} and parameters in Methods. 
\arthur{The inset shows a zoom-in representation of the corresponding inference graph of the power-grid model \eqref{eq.generator}--\eqref{eq.load} around an equilibrium point.}
We assume that the power grid is equipped with phasor measurement units (PMUs) randomly placed on a subset of load and generator terminal buses, comprising the set of sensor nodes $\mathcal S\subseteq \{\phi_{n_{\rm g}+1},\dots,\phi_N\}$.
The PMU measurements are transmitted to a control center in real time to support automated control actions, human decision-making, and cyber-attack detection.
We assume the system initially operates in steady state when, taking advantage of an otherwise inconsequential perturbation at time $t=1$~s (Fig.~\ref{fig.powergrid}\textit{B}), a deception cyber-attack tampers with the measured data $\phi_{\rm a}(t)$ from one of the sensors, transmitting instead false data $\tilde{\phi}_{\rm a}(t)$ to the control center for $t>1$ s. For illustration purposes, in Fig.~\ref{fig.powergrid}\textit{C} we assume that the false data are copied from the measurements of some neighboring vertex $j$, i.e., $\tilde{\phi}_{\rm a}(t)=\{\phi_j(t): K_{{\rm a}j}> 0\}$.

We show that this cyber-attack can be successfully detected by designing functional observers and cross-validating the transmitted measurements $\bar{\phi}_{\rm a}(t)$ against the state estimates $\hat{\phi}^{(\ell)}_{\rm a}(t)$ of each observer $\ell$. 
This cross-validation takes place during the short transient dynamics that follow the perturbation, where $\bar{\phi}_{\rm a}(t)\leftarrow\tilde{\phi}_{\rm a}(t)$ if there is an attack and $\bar{\phi}_{\rm a}(t)\leftarrow{\phi}_{\rm a}(t)$ otherwise.
Since one has no access to the true state estimation error $\phi_{\rm a}(t)-\hat{\phi}_{\rm a}^{(\ell)}(t)$, such cross-validation is performed statistically, relying on the state estimation of \textit{multiple} functional observers designed from distinct $\mathcal S^{(\ell)}\subset\mathcal S$, with cardinality $|\mathcal S^{(\ell)}|=|\mathcal S|/2$, as shown in Fig.~\ref{fig.powergrid}\textit{C}. To cross-validate the state estimates against the transmitted data, we use the root-mean-square error (RMSE) index:

\begin{equation}
    \operatorname{RMSE} = \sqrt{\frac{1}{t_{\rm d}} \int_{1}^{t_{\rm d}+1}\| \bar{\phi}_{\rm a}(t) - \hat{\phi}_{\rm a}^{(\ell)}(t)\|^2 dt}.
\end{equation}

\noindent
where $t_{\rm d}$ is the \textit{time of detection} window, defined as the time it takes to reliably detect an attack after it is launched. 
Clearly, the performance of the detection method depends on $t_{\rm d}$ (shaded window in Fig.~\ref{fig.powergrid}\textit{C}), which is also a lower bound of the time interval under which the system is left unprotected waiting for a decision.

Fig.~\ref{fig.powergrid}\textit{D} shows that, after a short period of time ($\approx$0.25~s), the separation between the histograms corresponding to the attacked and unattacked systems becomes statistically significant so that reliable detection can be made. Since the separation between the two histograms becomes more pronounced as time increases, the detection becomes more accurate when a larger detection window is allowed.
While the histograms show an aggregate representation of attacks under different perturbation scenarios, the same conclusions hold for individual attacks. In our simulations, reliable detection can be achieved even in the hardest-to-detect attack among all realizations, as illustrated by the contoured histogram for $t_{\rm d}=1$~s.
%
\arthur{The asymptotic convergence of the functional observer is guaranteed if the perturbed state remains in the attraction basin of the given equilibrium. For larger perturbations crossing into the basin of a different equilibrium, the functional observer has to be redesigned around the new equilibrium. 
Thus, the extent to which the designed observer remains valid upon large perturbations is ultimately determined by the basin stability of the nonlinear system, which can be assessed numerically as proposed in \cite{Menck2013,Menck2014}.}


In applications with a constantly changing operation point, such as smart power grids, algorithms for the design of controllers and observers have to be sufficiently fast so that they can be implemented in real time following a change of the equilibrium operation point.
Moreover, the statistical significance of the cyber-attack detection method increases with the number of implemented observers. Thus, the method can be used statistically only if, in addition, the algorithms for the observer design are fast enough to allow for a sufficiently large number of observers to be implemented in real time.
To that end, Algorithm~\ref{alg.findF0} provides a fast and scalable solution for the design of minimum-order functional observers. Fig.~\ref{fig.powergrid}\textit{E} 
shows for the IEEE-118 system that the functional observer usually has a much smaller dimension than the Luenberger observer as the number of PMUs increases, leading to a running time reduction by a factor of up to a hundred. This example illustrates for a small power grid the results anticipated in Figs.~\ref{fig.obsvorder} and \ref{fig.obsvestimation}, and we expect that the computational gain of the functional observers 
will increase as the network size increases (Fig.~\ref{fig.obsvestimation}\textit{D}).

\bigskip
\noindent\textbf{Estimation of epidemic spreading under limited testing.}
\label{sec.app.epidemics}
Motivated by the unprecedented impact of the COVID-19 pandemic, recent studies have highlighted the importance of epidemiological models for the design of containment measures.
Such models are useful for understanding the growth patterns and scaling laws governing the epidemic spreading \cite{Singer2020,Blasius2020} as well as for developing control strategies \cite{Lesniewski2020,Tsay2020,Morris2020}, which ultimately support policy-making decisions \cite{Hethcote2000}.
For instance, these models can inform decisions on social distancing and quarantine measures, which cannot be taken lightly as they generally involve social and economic costs.
The quality of the model predictions, and thus their ability to inform decisions, is strongly dependent on the state of the epidemic in a population, which is often only partially known due to limitations in testing and reporting. This is especially the case for a new and rapidly evolving pandemic, since it takes time to mobilize medical resources and ramp up testing capacity.
In particular, as illustrated in early stages of the COVID-19 pandemic, the testing capacity can vary widely across cities even within the same country.
Therefore, it is important to develop the ability to infer as much information as possible from the available incomplete data. Previous work has shown that state estimators can provide meaningful estimates of the true state of the number of infected, susceptible, and recovered individuals in an epidemic when sufficient data are available \cite{Iggidr2019,Tsay2020}. 

\begin{figure*}[t]
	\centering
	\includegraphics[width=\textwidth]{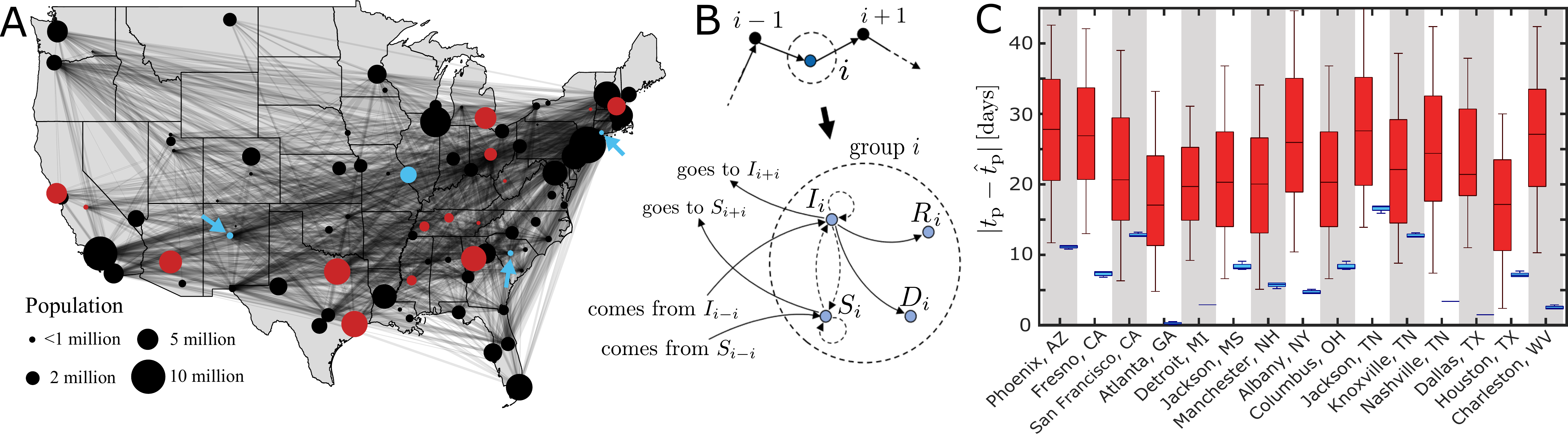}  
	\caption{
		Target state estimation in epidemics.
		(\textit{A})~U.S. air transportation network, where vertices represent cities and edges represent the direct flights. The target (red) and sensor (blue) cities are highlighted.
		(\textit{B})~Network flow between three cities (top) and the corresponding inference graph  (bottom) of the dynamical system \eqref{eq.sird}. The dynamics are taken into account by expanding each vertex $i$ as a set of SIRD state nodes, where links represent the linear (solid lines) and nonlinear (dashed lines) interactions between the state variables in the differential equations.
		(\textit{C})~Boxplot of the error between the time $t_{\rm p}$ of the epidemic peak in each target city and the predicted peak time $\hat{t}_{\rm p}$. The red bars show predictions from free-run simulations, while the blue bars show predictions given by the estimates of the designed functional observer, both for 100 independent realizations with arbitrary initial conditions.
		The bottom, middle and top of each box represent the 25th, 50th (median) and 75th percentiles of the sample, respectively; the whiskers mark the 5th and 95th percentiles.
		For illustration purposes, this example assumed 
		an outbreak initiated in Miami, FL, and its spreading dominated by domestic air transportation. 
        See Methods for simulation and modeling details.
	}
	\label{fig.epidemics}
\end{figure*}

Here, we show that, for target estimation, functional observers can be designed to provide reliable information using reduced testing data compared to full-state estimators. As the underlying dynamics in this case are inherently nonlinear, this will give us the opportunity to demonstrate that the methods established here for linear systems are also informative for nonlinear ones.  We illustrate this application for the estimation of the infected population in a set of ``target cities'' (where testing is inadequate) from the known case-fatality rate in a set of ``sensor cities'' (where sufficient testing is conducted). 
To this end, consider the multi-group 
model of the spreading of an infectious disease mediated by the air transportation network \cite{Colizza2006}:

\begin{equation}
\begin{aligned}
\begin{cases}
\dot S_i = \mathlarger{- \beta_i\frac{S_iI_i}{P_i} -\sum_{j=1,j\neq i}^N K_{ji}\frac{S_i}{P_i}+\sum_{j=1,j\neq i}^N K_{ij}\frac{S_j}{P_j},} \\
\dot I_i = \mathlarger{\beta_i\frac{S_i I_i}{P_i} - \gamma I_i -\sum_{j=1,j\neq i}^N K_{ji}\frac{I_i}{P_i}+\sum_{j=1,j\neq i}^N K_{ij}\frac{I_j}{P_j},} \\
\dot R_i = \mathlarger{(1-\eta)\gamma I_i,} \\
\dot D_i = \mathlarger{\eta\gamma I_i,}
\end{cases}
\end{aligned}
\label{eq.sird}
\end{equation}

\noindent
for $i=1,\ldots,N$, where $N$ is the number of groups, $(S_i,I_i,R_i,D_i)$ are the susceptible, infected, recovered, and dead (SIRD) individuals of group $i$ with population size $P_i=S_i+I_i+R_i+D_i$. Parameters $\gamma$ and $\eta$ are the recovery and fatality rate, respectively, $\beta_i$ is the contact rate in group $i$, and $K$ is the adjacency matrix of the transportation network, where $K_{ij}$ describes the number of individuals traveling from group $j$ to $i$ daily. As indicated in Fig.~\ref{fig.epidemics}\textit{A}, we assume each group to be a city in the United States and matrix $K$ to describe the air traffic between the cities' airports, while $(\beta_i,\gamma,\eta)$ and $K$ are chosen for illustration purposes to reflect the early stages of the COVID-19 epidemic (Methods). The epidemiological model \eqref{eq.sird} can be represented as an inference graph, where each state variable is a node and links represent linear and nonlinear interactions between variables, as illustrated by Fig.~\ref{fig.epidemics}\textit{B}.

Within the idealized model \eqref{eq.sird}, if the exact state of the epidemic is known at a given time, subsequent containment measures to ``flatten the curve'' could be designed based on a free-run simulation of the model. Unfortunately, this would not be the case in practice even if an accurate model were available because, as noted above, there are limitations on how precise are the data on the number of infected individuals in a population.
This is illustrated in Fig.~\ref{fig.epidemics}\textit{C} (red) for the predicted epidemic peaks in a set of 15 target cities.
We circumvent this limitation by designing a functional observer that provides more reliable estimates of the number of infected individuals.

A salient property of this problem is that the epidemiological model \eqref{eq.sird} is nonlinear. 
Our algorithms are not guaranteed to determine a \textit{minimum} set of sensors and a \textit{minimum}-order functional observer if the system is nonlinear. Notwithstanding, in many applications, Algorithms~\ref{alg.MSP} and~\ref{alg.findF0} can still be used to identify a small set of sensor nodes $\mathcal S$ (relative to the number of state variables) and a small order matrix $F_0$ (compared to that of a full-state observer) for the design of the corresponding \textit{nonlinear} functional observer. 
In particular, it can be shown that the nonlinear functional observer for system \eqref{eq.sird} satisfies the theoretical conditions for asymptotic convergence of the state estimates (Supporting Information, Section~\ref{sec.nonlinearfuncobsv}).

To design the functional observer, let a city $i$ be a sensor (target) city if $D_i\in\mathcal S$ (if $I_i\in\mathcal T$) is a sensor (target) node.
Given the specified set of target cities,  Fig.~\ref{fig.epidemics}\textit{A} highlights the selected set of 4 sensor cities for functional observability as provided by Algorithm~\ref{alg.MSP}. We then design the nonlinear functional observer based on the inference graph of system \eqref{eq.sird} and sets $\mathcal S$ and $\mathcal T$ using Algorithm~\ref{alg.findF0} (Supporting Information, Section~\ref{sec.nonlinearfuncobsv}).
Fig.~\ref{fig.epidemics}\textit{C} (blue) shows the estimated time of the epidemic peaks 
as provided by the designed nonlinear functional observer. Clearly, there is a great improvement in the estimation accuracy (i.e., a smaller deviation between realizations of different initial conditions), which highlights the observer's resilience to false initial predictions.
Note that this functional observer is designed in a situation where the system is unobservable but is functionally observable. This highlights a fundamental advantage of our approach: when the conventional full-state estimators are not applicable, our approach may still provide high-quality estimates of the nonlinear system's state from sparse measurement data.

\section*{Discussion}
\label{sec.discussion}

In large-scale complex networks, it often is physically impossible to ensure complete observability or computationally prohibitive to design full-state observers. This poses fundamental challenges to our ability to observe, understand, and control network processes. Yet, many practical applications only require the observation of a small subset of key variables, which we formalize by introducing the notion of structural functional observability. This work establishes graph-theoretical conditions for functional observability, enabling direct application to large-scale network systems. In particular, the resulting theory allows us to devise highly scalable algorithms to optimally allocate sensors and design functional observers for accurate estimation of a target subset of all state variables.

Observability and controllability are dual properties in control theory \cite{Kalman1959}, which leads to the natural question of what would be the dual of functional observability. 
%
By parallelism, 
one might be tempted to assume that a
notion of ``target observability'' for the given target set could be defined as the target controllability condition \cite{Gao2014} for the pair $(A^\transp,C^\transp)$. However, we argue that\textemdash in contrast with the target controllability condition for controller design\textemdash such notion of target observability does not lead to the design of an observer capable of estimating the state of the target nodes. Crucially, target estimation in this sense can only be accomplished with functional observability, which is a stronger condition than target observability. Indeed, based on the observability Gramian, we can show that the initial condition of the target nodes is uniquely reconstructable from measurements if and only if condition \eqref{eq.linearfuncobsv} is satisfied (Supporting Information, Section~\ref{sec.relatedworks}).
This leads us to a fundamental conclusion: \arthur{despite not being strictly dual,} functional observability and target controllability are \arthur{mirror properties} in terms of their functionalities.

Applications of the results presented here may include cyber-security in the decentralized control of infrastructure and multi-agent systems, state estimation for epidemic and ecosystem management from incomplete observation, and identification of biomarkers for prognosis, diagnosis, and treatment.
%
In biomedical and ecological applications, for example, assessment of the state of the system must not interfere significantly with the system dynamics in order to avoid mispredictions and adverse effects.
Our results provide a framework under which variables of interest can be estimated from indirect measurements, avoiding variables that interfere with the system's function and control actions. They also provide an alternative to reduce the computational costs associated with sensing, communication, and data processing in infrastructure, supply, and technological networks that require real-time feedback control \cite{WangMorse2018}.

This work also leads to fundamental questions worth pursuing in future research.
First, because our theory allows a decision on whether the system is functionally observable or not from the graph structure of the model alone, it is natural to consider how to design functional observers when specific parameters of the system are unknown. This could be addressed by combining our graph-based methods with machine learning techniques to enable data-driven state estimation by functional observers.
\arthur{Second, while here we considered the sensor placement problem and functional observer design in a pre-existing cyber-physical network, combining these approaches with the co-design of the communication (and/or physical) layers of the system can lead to further improved resilience against attacks and failures.}
\arthur{In particular, the transition to renewable energy is leading to highly heterogeneous networks that integrate conventional synchronous generators, converter-based distributed generation units (solar and wind), flexible AC transmission systems, and high-voltage DC transmission systems. To ensure the functional observability and stable operation of future power systems, detailed models of these critical components \cite{vittal2019power} can be incorporated in the co-design of the communication network for wide-area monitoring, protection, and control systems \cite{vaccaro2016wide}.}
Third, it would be interesting to examine the estimation accuracy, convergence rate, and stability of functional observers in large-scale systems under modeling uncertainties, measurement noise, and round-off errors.
\arthur{A systematic study of the theoretical performance of functional observers in the presence of bounded modeling and measurement errors is still an open problem in the literature, which may be approached in a framework dual to the problem of robust control \cite{zhou1998essentials}.}
Finally, our application example to epidemics illustrates how the modeling of the system structure can be used to design functional observers in nonlinear systems. Deriving conditions that allow our methods to be extended to more general nonlinear systems with guarantees of optimality is left for future work.

\section*{Methods}

\noindent\textbf{Observer design.}
Both the Luenberger and the functional observer are auxiliary dynamical systems that can be designed under the following structure:
\begin{equation}
\begin{aligned}
\begin{cases}
	\dot{\bm w} = N\bm w + J\bm y + H\bm u, \\
	\hat{\bm z} = D\bm w + E\bm y,
\end{cases}
\end{aligned}
\label{eq.functionalobserverdyn}
\end{equation}

\noindent
where $(N,J,H,D,E)$ are design matrices with consistent dimension. In a properly designed Luenberger observer, $\hat{\bm z}$ converges asymptotically to the $(n-q)$ unmeasured states of the system \eqref{eq.linearsys}, while, in a properly designed functional observer, $\hat{\bm z}$ converges asymptotically to $\bm z_0 = F_0\bm x$. The initial state $\bm w(0)$ of the observer can be assigned arbitrarily and the initial state $\bm x(0)$ is unknown. 

Despite sharing the same equations \eqref{eq.functionalobserverdyn}, the Luenberger and functional observers have \textit{very} different orders (defined by the dimension of vector $\bm w$)
and involve \textit{very} different design procedures. In what follows, the design procedure for each observer is presented using 
the linear transformation
\begin{equation}
P^{-1}AP = \begin{bmatrix}
A_{11} & A_{12} \\ A_{21} & A_{22}
\end{bmatrix},
\,\,\,
PB = \begin{bmatrix} 
B_1 \\ B_2 
\end{bmatrix},
\,\,\,
F_0P = \begin{bmatrix}
F_1 & F_2
\end{bmatrix},
\label{eq.lineartransf}
\end{equation}

\noindent
where $P=[C^\dagger \,\,\, C^\perp]$ is the transformation matrix, $C^\dagger$ is the Moore-Penrose inverse of $C$, and $C^\perp$ is the orthogonal complement of the row space of $C$.

If $(A,C)$ is observable, then a stable Luenberger observer (with arbitrary stable poles) can be designed by defining $N = A_{22}-EA_{12}$, $J=A_{21}-EA_{11} + NE$, $H = B_2 - EB_1$, and $D=I_{n-q}$, where $E$ is a design matrix that can be found using any pole-placement algorithm such that $(A_{22}-EA_{12})$ is Hurwitz. Note that $\bm w\in\R^{n-q}$.

If $(A,C,F)$ is functionally observable and $(A,C,F_0)$ satisfies conditions \eqref{eq.darouachcond1}--\eqref{eq.darouachcond2} for some $\operatorname{row}(F_0)\supseteq\operatorname{row}(F)$, then a stable functional observer (with arbitrary stable poles) can be designed as follows \cite{Darouach2000,Trinh2012}: i)~Compute $N_1=(\Phi\Omega^\dagger A_{12}+F_2A_{22})F^\dagger_2$ and $N_2=(\Omega\Omega^\dagger - I_q)A_{12}F_2^\dagger$, where $\Omega = A_{12}F_2^\perp$ and $\Phi = -F_2A_{22}F_2^\perp$; ii)~Find $Z$ using any pole-placement algorithm such that $N = N_1 - ZN_2$ is Hurwitz; iii)~Compute $T = [T_1 \,\,\, T_2]$, where $T_1 = \Phi\Omega^\dagger + Z(I_q-\Omega\Omega^\dagger)$ and $T_2 = F_2$; iv)~Compute $D = I_r$, $J = T_1A_{11} + T_2A_{21}-NT_1$, $H=TB$, $E = F_1 - DT_1$. Note that $\bm w\in\R^{r_0}$.

In this paper, we use the linear-quadratic regulator (LQR) as a pole-placement algorithm, which requires solving the algebraic Riccati equation $X^\transp P + PX - PYR^{-1}Y^\transp P + Q = 0$ for $P$. 
For the Luenberger observer design, let $E\leftarrow P$, $X\leftarrow A_{22}^\transp+\alpha I$, and $Y\leftarrow A_{12}^\transp$. 
For the functional observer design, let $Z\leftarrow P$, $X\leftarrow N_2^\transp+\alpha I$, $Y\leftarrow N_1^\transp$.
In both cases, we define $\alpha = -100$, $Q = 10^{-3}\cdot I$, and $R=I$,  which leads to observers with
minimum estimation energy ($R\gg Q$). The diagonal terms in $X$ guarantee that $Z$ and $E$ are designed to have the right-most eigenvalues equal to $\alpha$. This ensures that their dynamics are dominated by the same slowest eigenvalue, allowing a consistent comparison of the observers' performance despite their different orders.

\bigskip
\noindent\textbf{Minimum sensor placement algorithm.}
Algorithm \ref{alg.MSP} provides an approximate solution to the minimum sensor placement problem in polynomial time.
The key steps are as follows.
First, a breadth-first search is run for each target node (for-loop), allowing us to determine the sets of target nodes $\mathcal R_i\subseteq\mathcal T$ that have a direct path to each state node $x_i\in\mathcal K\subseteq\mathcal X$ in $\mathcal G(A)$, where $\mathcal K$ is a set of candidate nodes for sensor placement. 
Note that a breadth-first search has a complexity of order $O(n+|\mathcal E|)$ \cite{Newman2010},
where $|\mathcal E|$ is the cardinality of the set of links $\mathcal E$ in $\mathcal G(A)$, and can be run in parallel for each $x_i\in\mathcal T$. 
Second, a greedy algorithm is used (while-loop) to find an approximation of the minimum set of sensor nodes such that structurally functional observability is guaranteed. In the worst-case scenario, the greedy search has a complexity of order $O(n^2)$.

\begin{algorithm}[H]
	\caption{\label{alg.MSP} \small Minimum sensor placement}
	\begin{flushleft}
		
		\textbf{input:} inference graph $\mathcal G(A^\transp)$, target set $\mathcal T$, candidate set $\mathcal K$
		
		\textbf{output:} sensor set $\mathcal S$
		
		\smallskip
		
		initialize $\mathcal R_i\leftarrow\emptyset$, $\forall i=1,\ldots,|\mathcal K|$;
		
		\textbf{for} all $x_i\in\mathcal T$
		
		\begin{itemize}
			\vspace{-4pt}
			
			\item[] starting at node $x_i$ in graph $\mathcal G(A^\transp)$, find the set of reachable nodes $\mathcal R'_i\subseteq\mathcal X$ using a breadth-first search algorithm;
			
			\item[] \textbf{for} all $x_j\in\mathcal K$
			
			\begin{itemize}
			    \vspace{-4pt}
			    
			    \item[] if $x_j\in\mathcal R'_i$, then $\mathcal R_j\leftarrow\mathcal R_j\cup\{x_i\}$; 
			    
			    \vspace{-4pt}
			    
			\end{itemize}
			
			\item[]\textbf{end}
		
			\vspace{-4pt}
		\end{itemize}
	
		\textbf{end}
		
		\smallskip
		
		initialize $\mathcal S\leftarrow\emptyset$.

		\textbf{do}
		\begin{itemize}
			\vspace{-4pt}

			\item[] for all $x_i\in\mathcal K\backslash\mathcal S$, compute gain
			\begin{equation}
			\Delta(x_i) = \Big|\bigcup_{j : x_j\in\mathcal S\cup\{x_i\}}\mathcal R_j\Big| - \Big|\bigcup_{j : x_j\in\mathcal S}\mathcal R_j\Big|;
			\end{equation}
			
			\item[] add the element with the highest gain
			\begin{equation}
			\mathcal S \leftarrow \mathcal S\cup\{\arg\max_{x_i}\Delta(x_i)|x_i\in\mathcal K\backslash\mathcal S\};
			\label{eq.greedygain}
			\end{equation}
			
			\vspace{-4pt}
		\end{itemize}
	
		\textbf{while} $\bigcup_{j:x_j\in\mathcal S}\mathcal R_j \neq \mathcal T$.
	\end{flushleft}
\end{algorithm}

\bigskip
\noindent\textbf{Minimum-order functional observer design algorithm.}
For cases where $(A,C,F)$ is functionally observable, Algorithm \ref{alg.findF0} provides a scalable solution to the problem of determining $F_0$ with \textit{minimum} order such that condition \eqref{eq.darouachcond1} is satisfied for the triple $(A,C,F_0)$ (Supporting Information, Corollary \ref{corol.algminF0}). 
In this algorithm, we avoid numerical computation of the rank condition in \eqref{eq.darouachcond1}, which is numerically unstable and computationally demanding for high-dimensional matrices (e.g., numerical rank computation based on SVD methods has a complexity of order $O(n^3)$). Instead, we compute the structural (or generic) rank of a matrix by finding the maximum matching of the corresponding \textit{bipartite graph} of such a matrix. This is a highly scalable alternative since the maximum matching problem can be solved by the Hopcroft-Karp algorithm, which has a complexity of order $O(\sqrt{n_b}|\mathcal E_b|)$, where $n_b$ and $|\mathcal E_b|$ are the numbers of nodes (columns and rows) and links (non-zero entries) in the bipartite graph (matrix).
Fig.~\ref{fig.algorithm} (Supporting Information) presents an illustrative example of Algorithm~\ref{alg.findF0}, where it becomes clear how we take advantage of the structural properties of a dynamical system to augment $F_0$ at every iteration until condition \eqref{eq.darouachcond1} is satisfied. 
%

Algorithm \ref{alg.findF0} finds the minimum-order $F_0$ in $O((r_0-r)\sqrt n_b|\mathcal E_b|)$ time, where $n_b$ and $|\mathcal E_b|$ are the numbers of nodes and links in $\mathcal B$, respectively. This follows from $F_0$ being determined with $(r_0-r)$ recursive iterations, in which a maximum matching algorithm of $O(\sqrt{n_b}|\mathcal E_b|)$ is run at each iteration. This complexity order can be estimated as a function of $n$ in a worst-case scenario in which one has a single sensor node ($q=1$) and a single target node ($r=1$), but\textemdash in order to satisfy \eqref{eq.darouachcond1}\textemdash all other unmeasured nodes must be estimated, resulting in $r_0 = n-q \approx n$ and hence $r_0-r\approx n$. Since there are at most $n_b=2q+r_0+n\approx 2n$ nodes in $\mathcal B$, let $|\mathcal E_b|=n_bk_{\rm avg}$, where $k_{\rm avg}$ is the average node degree in $\mathcal B$. Thus, the complexity order is $O(n^{2.5})$ if we assume that $k_{\rm avg}\ll n$. Note that, being a worst-case scenario, this is still a very conservative estimate since usually $r_0\ll n$ (Fig.~\ref{fig.obsvorder}). This estimate should be contrasted with the complexity order of $O(n^4)$ in the worst-case scenario for the procedure in \cite{Fernando2010b}, which requires computation of the rank condition via SVD in approximately $n$ iterations (Supporting Information, Section~\ref{sec.graphalgorithm}).

\begin{algorithm}[H]
	\caption{\label{alg.findF0}\small Minimum-order functional observer design}
	\begin{flushleft}
			
		\textbf{input:} functionally observable triple $(A,C,F)$
		
		\textbf{output:} functional observer matrices $(F_0,N,J,H,D,E)$
		
		\smallskip
		
		initialize $F_0\leftarrow F$, $r_0 \leftarrow \rank(F_0)$, $\mathcal M_1\leftarrow\emptyset$, $\mathcal M_2\leftarrow\emptyset$;
		
		\textbf{do}
		
		\begin{itemize}
			\vspace{-6pt}
			
			\item[] update $G \leftarrow [C^\transp \,\,\, (CA)^\transp \,\,\, F_0^\transp]^\transp$;
			
			\item[] build a bipartite graph $\mathcal B(\mathcal V,\mathcal V',\mathcal E_{\mathcal V,\mathcal V'})$, where $\mathcal V=\{v_1,\ldots,v_{2q+r_0}\}$ is a set of nodes with each element corresponding to a row of $G$, $\mathcal V'=\{v'_1,\ldots,v'_n\}$ is the set of nodes with each element corresponding to a column of $G$, and $\mathcal E_{\mathcal V,\mathcal V'}$ is the set of undirected links $(v_i,v'_j)$ defined by the non-zero entries $G_{ij}$ of $G$;
			
			\item[] find the maximum matching set $\mathcal E_{m}$ associated with $\mathcal B(\mathcal V,\mathcal V',\mathcal E_{\mathcal V,\mathcal V'})$ (e.g., via the Hopcroft-Karp algorithm);
			
			\item[] for all $v'_i\in\mathcal V'$, if $v'_i$ and each of its second neighbors are connected to a link in $\mathcal E_m$, then update the set of right-matched nodes $\mathcal M_1\leftarrow\mathcal M_1\cup\{v'_i\}$;
			
			\item[] define the set of candidate nodes $\mathcal K = \mathcal M_2\backslash\mathcal M_1$, where $v'_j\in\mathcal M_2$ if $[F_0A]_{ij}$ is a non-zero entry for some $i$;
			
			\item[] draw  an element $v'_k\in\mathcal K$ 
			and update $F_0\leftarrow[F_0^\transp \,\,\, (F')^\transp]^\transp$ ($r_0\leftarrow r_0+1$), where $F'\in\R^{1\times n}$ and $[F']_{1j}=1$ if $j=k$, and $[F']_{1j}=0$
			otherwise; 			
			
		\end{itemize}
		
		\textbf{while} $\mathcal K \neq \emptyset$;
		
		compute the functional observer matrices $(N,J,H,D,E)$ for a triple $(A,C,F_0)$.
	\end{flushleft}
\end{algorithm}

\bigskip
\noindent\textbf{Optimal sensor placement for minimum-order functional observers.} Given a functionally observable triple $(A,C,F)$, one may be interested in how to optimally place additional sensor nodes in a network such that the functional observer order $r_0$ is minimized.
This is a bi-level optimization problem
\begin{equation}
\min_{
\mathcal S\subseteq\mathcal K\backslash\mathcal T, |\mathcal S| \leq q
} \ \min_{
F_0\in\{0,1\}^{n\times n}
}
J(\mathcal S, F_0), \,\, {\rm s.t. \ conditions} \,\, \eqref{eq.linearfuncobsv} \text{-} \eqref{eq.darouachcond1},
\label{eq.bilevelopt}
\end{equation}

\noindent
where $\mathcal K$ is a set of candidate nodes for sensor placement and $J(\mathcal S,F_0)$ is a cost function that returns the minimum order $r_0=\rank(F_0)$ of a functional observer. However, finding $F_0$ depends on $\mathcal S$ (which defines matrix $C$ in conditions \eqref{eq.linearfuncobsv}--\eqref{eq.darouachcond1}) and is ``embedded'' in a lower-level optimization task that requires, for instance, use of Algorithm \ref{alg.findF0}.  This is a hard-to-solve problem but, for illustration purposes, in Fig. \ref{fig.obsvorder}\textit{F} we implement a (non-scalable) greedy algorithm that recursively adds elements to $\mathcal S$ by letting $\mathcal S\leftarrow\mathcal S\cup\{\arg\max_{x_i}\Delta(x_i)|x_i\in\mathcal K\backslash\mathcal S\}$ until $|\mathcal S|=q$, where $\Delta(x_i) = J(\mathcal S\cup\{x_i\}) - J(\mathcal S), \forall x_i\in\mathcal K\backslash\mathcal S$, and $J(\mathcal S)$ is the order of $F_0$ returned by Algorithm \ref{alg.findF0} for this set $\mathcal S$.

\bigskip
\noindent\textbf{Generation of complex dynamical networks.} For the generation of the $N$-vertex complex networks used in Figs.~\ref{fig.msp}--\ref{fig.obsvestimation}, we explore the following parameters: $m\in\{1,2,\ldots,7\}$ for Barabási-Albert SF networks \cite{Barabasi1999}; $k=2$ and $p\in[0, 1]$ for Newman-Watts SW networks \cite{Newman1999}.
Parameter $m$ is the number of edges of each vertex iteratively added to the network, $k$ is the number of nearest neighbors in a ring graph, and $p$ is the probability of adding a new edge. 
For each of these undirected networks, a \textit{directed} model is generated by randomly assigning a single direction to each edge.

Since we are concerned with \textit{dynamical networks}, we consider that each vertex in an SW or SF network is a dynamical system of possibly more than one dimension. For illustration purposes, we assume that each vertex of a network generated from these models is a 3-dimensional subsystem with the following general structure:
\begin{equation}
A_{\rm v} = \begin{bmatrix} -1 & -1 & 0 \\ 1 & -1 & 0 \\ 1 & 0 & -1 \end{bmatrix}.
\label{eq.nodalstructure}
\end{equation}

\noindent
To include the effects of heterogeneity in the vertex dynamics of the dynamical networks generated, we let the dynamics of each subsystem be defined by $A_i = \lambda_iA_{\rm v}$ for $\lambda_i\sim\mathcal U[2,5]$, where $i=1,\dots,N$ and $\mathcal U[a,b]$ is a uniform distribution in the interval $[a,b]$.
Thus, the system matrix $A$ describing the whole dynamical network is given by $A = {\rm diag}(\lambda_1,\ldots,\lambda_N)\otimes A_{\rm v} - L\otimes M$, where $\otimes$ is the Kronecker product operator, $L$ is the Laplacian matrix of the generated network, and $M\in\{0,1\}^{3\times 3}$ is defined by $M_{ij} = 1$ if $i=j=2$ and $0$ otherwise. The term $L\otimes M$ means that the second state variable of all subsystems are diffusively coupled according to $L$. Note that $A$ has dimension $n=3N$, which is also the number of nodes in $\mathcal G(A)$.

The generalized clustering of the inference graph $\mathcal G(A)$, studied in Fig.~\ref{fig.obsvorder}\textit{D}, is defined as $\operatorname{C_{\rm g}} = \frac{1}{n}\sum_{i,j,k} {A'_{ij} A'_{jk} A'_{ki}}/{k_i^2}$, where $k_i$ is the node degree of node $x_i$ in $\mathcal G(A')$ and $A' = ({A+A^\transp})/{2}$. 
Note that the diagonal entries $a_{ii}$ are included to account for self-edges, which distinguishes $C_{\rm g}$ from the standard clustering coefficient, and that the computation is effectively done for the undirected counterpart of $\mathcal G(A)$.

For the results presented in Figs.~\ref{fig.obsvorder}--\ref{fig.obsvestimation}, 
the target and sensor nodes were chosen randomly among the first state variable of each subsystem (vertex) $A_i$, under the condition that
the sets of sensor nodes and target nodes are non-overlapping (i.e., $\mathcal S\cap\mathcal T = \emptyset$). That is, $F_{ij}$ or $C_{ij}$ is a non-zero entry only if $(j+2)/3$ is integer, and thus
%
$\mathcal S\cup\mathcal T$ can have at most $N$ state nodes.
This procedure is relaxed in Fig.~\ref{fig.msp}, where all state variables are candidates for sensor placement (i.e., $\mathcal K = \mathcal X$).

\bigskip
\noindent\textbf{Real-world networks datasets.} For the real-world networks used in Fig.~\ref{fig.msp}\textit{B}, we take several adjacency matrices $A_{\rm adj}$ available in different real-world datasets (described in Supporting Information, Section~\ref{sec.rwnetsdatabase}). For each real-world network, we define the system matrix $A$ as the Laplacian matrix of $A_{\rm adj}$ in order to model the energy/information flow in such systems as diffusive processes.

\bigskip
\noindent\textbf{Parameters of the IEEE-118 power-grid model.}  
%
\arthur{The IEEE-118 system is derived from a representation of the U.S. Midwest system with 118 buses, 54 generators, and 186 lines \cite{Zimmerman2011}. The network parameters $K_{ij}$ and generator's mechanical power $P_i$ are computed from the nominal AC power flow solution, the generator and load dynamical parameters $H_i$ and $D_i$ are estimated following the method provided in \cite{Nishikawa2015}, and the nominal frequency $w_{\rm R}$ is 60~Hz.}
The power flow equations were numerically solved using the MATPOWER toolbox \cite{Zimmerman2011}. The initial conditions were set, for all $i$, assuming that the power system is in a synchronized steady state $\dot{\phi}_i(0) = 0$, with $\phi_i(0)$ determined by the power flow solution.

\bigskip
\noindent\textbf{Parameters of the epidemic spreading model.} The parameters in \eqref{eq.sird} were set as $(\beta,\gamma,\eta)=(0.4,0.16,0.01)$ in order to mimic the coronavirus spreading in each group according to results reported in \cite{Blasius2020}, where here $\beta$ is the average value of $\beta_i$. We define the contact rate of each group as $\beta_i\sim\mathcal N(\beta,0.01)$ to account for group heterogeneity in the simulation. The air transportation network defined by $K$ describes the traffic of passengers according to the TranStats database for international and domestic flights (\url{http://www.transtats.bts.gov/}). Multiple airports belonging to the same city are combined into a single group (vertex) $i$ with a population $P_i$ given by the corresponding city's population. In Fig.~\ref{fig.epidemics}, the epidemic peak time $t_{\rm p}$ is determined by numerically integrating \eqref{eq.sird} with the initial conditions $I_j(0)=10^3$ and $S_j(0) = P_j-10^3$, if $j$ is the index for Miami, FL, and $I_j(0) = 0$ and $S_j(0) = P_j$, otherwise. The predicted peak time $\hat t_{\rm p}$ is determined based on the free-run simulations and functional observer estimations, both initialized with a false guess of the outbreak in a randomly selected city $k$. That is, $I_j(0) = 1$ and $S_j(0) = P_j-1$, if $j=k$, and $I_j(0) = 0$ and $S_j(0) = P_j$, otherwise.

\bigskip
\noindent\textbf{Data Availability.}
Codes and data have been deposited in GitHub (\url{https://github.com/montanariarthur/FunctionalObservability}).

\begin{acknowledgments}
The authors acknowledge support from Brazil's Coordena\c{c}\~ao de Aperfei\c{c}oamento de Pessoal de N\'ivel Superior (finance code 001) (A.N.M.) and Conselho Nacional de Desenvolvimento Cient\'ifico e Tecnol\'ogico (Grant 03412/2019-4) (L.A.A.), and U.S. Army Research Office (Grant W911NF-19-1-0383) (C.D. and A.E.M.).
\end{acknowledgments} 


%


\balancecolsandclearpage
\clearpage

\begin{center}
    \Large\textbf{Supporting Information}
\end{center}

\bigskip
\noindent\textbf{Notation}
\label{sec.nomenclature}
Throughout the main and supplementary text, we adopt the following nomenclature and notation. Column vectors are represented by bold lower-case letters (e.g., $\bm x, \bm y, \bm z$), matrices are represented by capital letters (e.g., $X, Y, Z$), and sets are represented by calligraphed capital letters (e.g., $\mathcal X, \mathcal Y, \mathcal Z$). 
The exceptions to this notation are the observability and the controllability matrices, which are denoted by $\mathcal O$ and $\mathcal C$, respectively, to avoid notation conflict with the complexity of order $O(\cdot)$ and the output matrix $C$.
$|\mathcal V|$ is the cardinality (number of elements) of a set $\mathcal V$.
$I_n$ denotes the identity matrix of size $n$ and $\it 0_{m\times n}$ is an $m\times n$ null matrix (the subscript is often omitted when self-evident from the context). 
$A^\perp$ is the orthogonal complement of the row space of $A$.
$A^\dagger$ is the Moore-Penrose inverse of $A$ (i.e., $AA^\dagger = I_p$).
${\rm row}(A)$ is the row space of $A$.
${\mathcal A}\subseteq {\mathcal B}$ indicates that ${\mathcal A}$ is a subset of ${\mathcal B}$, and ${\mathcal A}\subset {\mathcal B}$ further indicates a proper subset.
${\rm spec}(A)$ is the spectrum of matrix $A$. 
$\operatorname{diag}(A_1,\ldots,A_n)$ is a block-diagonal matrix formed by matrices $A_1,\ldots,A_n$. 
$[A]_{ij}$ and $A_{ij}$ are used to denote the $(i,j)$-th entry of matrix $A$. 
The operator $\otimes$ indicates the Kronecker product. 
$x\sim N(\mu,\sigma)$ is a realization of a random variable drawn from a Gaussian distribution with mean $\mu$ and standard deviation $\sigma$. 
$x\sim\mathcal U[a,b]$ is a realization of a random variable drawn from a uniform distribution in the interval $[a,b]$.

\section{Background on structural observability}
\label{sec.background}

We first recall Lin's definition for structural controllability \cite{Lin1974a} and observability \cite{Willems1986}. Let a matrix $A\in\{0,\star\}^{n\times n}$ be called a structured matrix if $A_{ij}$ is either a fixed zero entry or an independent free parameter, denoted by a $\star$. A matrix $\tilde A$ is a numerical realization of $A$ if real numbers are assigned to all free parameters of $A$. 

\begin{defin}
	The structured pair $(A,C)$ (or pair $(A,B)$)
	is structurally observable (controllable) if and only if there exists some numerical realization $(\tilde A,\tilde C)$ (or $(\tilde A,\tilde B)$) that is observable (controllable).
	\label{def.obs.linsdefinition}
\end{defin}


\begin{remark}
	Note that $\rank(\tilde A)\leq\rank(A)$. This upper bound is also known as structural or generic rank.
	\label{remark.structuralrank}
\end{remark}

The definition of structural observability  has a strong graph-theoretic interpretation. In order to interpret the results below, we present additional notation and definitions.

The corresponding \textit{inference graph} of a dynamical system $(A,B,C)$ is denoted by $\mathcal G(A,B,C)=\{\mathcal V,\mathcal E \}$, where $\mathcal V = \mathcal X\cup\mathcal U\cup\mathcal S$ and $\mathcal E = \mathcal E_{\mathcal X}\cup\mathcal E_{\mathcal U}\cup\mathcal E_{\mathcal S}$. The node sets are the state variables $\mathcal X = \{x_1,\ldots, x_n\}$, the input variables $\mathcal U = \{u_1,\ldots,u_p\}$, and the output variables $\mathcal S = \{y_1,\ldots, y_q\}$. A link $(x_i,x_j)$ (directed arrow from $x_j$ to $x_i$) is an element of $\mathcal E_{\mathcal X}$ if $A_{ij}$ is a free parameter entry of the structured matrix $A$, a link $(x_i,u_j)$ is an element of $\mathcal E_{\mathcal U}$ if $B_{ij}$ is a free parameter entry of the structured matrix $B$, and a link $(y_i,x_j)$ is an element of $\mathcal E_{\mathcal S}$ if $C_{ij}$ is a free parameter entry of the structured matrix $C$. For brevity, when studying the observability property of a pair $(A,C)$ we refer to the corresponding inference graph simply as $\mathcal G(A,C)=\{\mathcal X\cup \mathcal S, \mathcal E_{\mathcal X}\cup\mathcal E_{\mathcal S}\}$. Likewise, we use $\mathcal G(A,B)=\{\mathcal X\cup \mathcal U, \mathcal E_{\mathcal X}\cup\mathcal E_{\mathcal U}\}$  when considering the controllability property of a pair $(A,B)$.

A subset of nodes $\mathcal V'\subseteq\mathcal X$ has a \textit{dilation}\footnote{Note that this definition of dilation is dual to Lin's definition \cite{Lin1974a}, as here we focus on structural observability rather than controllability.} in an inference graph $\mathcal G(A,C)$ if and only if $|T(\mathcal V')| < |\mathcal V'|$, where $T(\mathcal V')$ is the set of all nodes $v_i\in\mathcal X\cup\mathcal S$ with the property that there is a direct link from a node in $\mathcal V'$ to $v_i$.
Let $\mathcal D_k$ be a \textit{minimal dilation set} of $\mathcal G(A,C)$, that is, a set with the property that, for all $\mathcal D_k'\subset \mathcal D_k$, the subset $\mathcal D_k'$ has no dilations. Let $\mathcal D = \bigcup_k \mathcal D_k$ be the union of all minimal dilation sets of $\mathcal G(A,C)$.

\begin{theor}
	\cite{Lin1974a}  The pair $(A,C)$ is structurally observable if and only if the corresponding inference graph $\mathcal G(A,C)$ satisfies the following conditions:
	\begin{enumerate}
		\item every state variable $x_i\in\mathcal X$ has a path to some output variable $y_i\in\mathcal S$;
		\item $\mathcal G(A,C)$ has no dilations.
	\end{enumerate}
	\label{theor.structuralobsv}
\end{theor}

\begin{remark}
	An inference graph $\mathcal G(A,C)$ has a dilation if there exists a set of $k$ columns of $[A^\transp \,\,\, C^\transp]^\transp$ which contains nonzero entries in less than $k$ rows of the submatrix formed by these $k$ columns. In fact, $\mathcal G(A,C)$ has a dilation if and only if $\rank [A^\transp \,\,\, C^\transp]^\transp < n$ \cite{Brisk1984}.
	\label{remark.dilationrankdeff}
\end{remark}

\section{Structural functional observability}
\label{sec.structfuncobsv}

We generalize the concept of structural observability to structural functional observability. 
For all subsequent
definitions and theorems, consider that the following assumptions hold.

\begin{assump}
	Let $F=[F_1^\transp \,\, \ldots \,\, F_r^\transp]^\transp$, where each row $F_i\in\R^{1\times n}$ has $n-1$ zero entries and only a single nonzero entry (a free parameter entry). Let $\mathcal T\subseteq\mathcal X$ be a set of state variables that we wish to estimate, also referred to as ``target nodes,'' where $x_j\in\mathcal T$ if and only if the $j$-th entry of some row $F_i$ is a nonzero entry. In other words, state variables are targeted (i.e., sought to be estimated) \textit{independently}. 
	\label{assump.independentF}
\end{assump}

\begin{assump}
	Let $\rank[C^\transp \,\,\, F^\transp]^\transp=q+r$, where $\rank(C)=q$ and $\rank(F)=r$.
	\label{assump.independentCF}
\end{assump}

\begin{remark}
	 Assumption \ref{assump.independentCF} involves no loss of generality since the state of any node in the $\operatorname{row}([C^\transp \,\,\, F^\transp]^\transp)$ can be determined as a linear combination of outputs and estimated target nodes, hence not requiring estimation.  
\end{remark}

\begin{defin}
	The structured triple $(A,C,F)$ is structurally functionally observable if and only if there exists some numerical realization $(\tilde A,\tilde C, \tilde F)$ that is functionally observable.
	\label{def.structuralfuncobsvdefin}
\end{defin}

\begin{theor}
	A triple $(A,C,F)$ is structurally functionally observable if and only if the corresponding inference graph $\mathcal G(A,C)$ satisfies the following conditions:
	\begin{enumerate}
		\item every state variable $x_i\in\mathcal T$ has a path to some output variable $y_i\in\mathcal S$;
		\item $\mathcal T\cap\mathcal D=\emptyset$, where $\mathcal D$ is the union of all minimal dilation sets of $\mathcal G(A,C)$.
	\end{enumerate}
	\label{theor.structuralfuncobsvdefin}
\end{theor}

\begin{proof} 
	This proof follows from the five following lemmas.
	Lemma \ref{lemma.somelambda} and \ref{lemma.lambda2} establish the necessity of conditions 1 and 2, respectively. Lemmas \ref {lemma.lambda4} and \ref{lemma.lambda5}, together, establish the sufficiency of conditions 1 and 2. The proof of Lemma \ref {lemma.lambda4} is based on borrowing the basic idea from \cite[Theorem C.3]{Sundaram2012} and extending it to functional observability. Lemma \ref{lemma.matrixpencil} is used in the proof of Lemma \ref{lemma.lambda4}.
	
	We remind the reader that an equivalent condition to \eqref{eq.linearfuncobsv}
	for functional observability is \cite{Jennings2011}  
	\begin{equation}
		\rank\begin{bmatrix} A - \lambda I \\ C \\ F \end{bmatrix}
		=
		\rank\begin{bmatrix} A - \lambda I \\ C	\end{bmatrix}
		\label{eq.linearfuncobsv2}
	\end{equation}
	\noindent for all $\lambda\in \mathbb{C}$. Since condition \eqref{eq.linearfuncobsv2} holds trivially for 
	$\lambda\notin{\rm spec}(A)$, we only need to consider the cases in which $\lambda \in{\rm spec}(A)$ to establish the equality for a given triple $(A,C,F)$.
\end{proof}

\begin{lemma}[Necessity of condition 1]
	If there exists at least one state variable $x_i\in\mathcal T$ that does not have a path to some output variable $y_i\in\mathcal S$, then, for any independent choice of free parameters in the system matrices $(A,C,F)$, there is at least one $\lambda\in\C$ such that $\rank\begin{bmatrix} A^\transp-\lambda I & C^\transp & F^\transp\end{bmatrix}^\transp>\rank\begin{bmatrix} A^\transp-\lambda I & C^\transp\end{bmatrix}^\transp$.
	\label{lemma.somelambda}
\end{lemma}

\begin{proof}
	Suppose that some nodes in $\mathcal X$ do not have a path to some output node in $\mathcal S$. Let $\mathcal X_1$ be the set of state variables that have a path to some output variable, and $\mathcal X_2 = \mathcal X\backslash \mathcal X_1$ denote all state variables that do not have a path to any output variable. Let $|\mathcal X_1|=k$ and $|\mathcal X_2|=n-k$. After applying a permutation of coordinates such that the nodes in $\mathcal X_1$ appear first, the matrices $(A,C,F)$ have the form
	\begin{equation}
	A = \begin{bmatrix} A_{11} & \it 0 \\ A_{21} & A_{22}\end{bmatrix},
	C = \begin{bmatrix} C_1 & \it 0 \end{bmatrix},
	F = \begin{bmatrix} F_1 & F_2 \end{bmatrix},
	\end{equation}
	
	\noindent
	where $A_{11}\in\R^{k\times k}$, $A_{21}\in\R^{(n-k)\times k}$, $A_{22}\in\R^{(n-k)\times (n-k)}$, $C_1\in\R^{q\times k}$, $F_1\in\R^{r\times k}$, $F_2\in\R^{r\times (n-k)}$, and $\it 0$ indicates a null vector or matrix of appropriate dimension. Note that there are no links from a node in $\mathcal X_2$ to $\mathcal X_1\cup\mathcal S$. Thus, we have the following matrix pencil:
	\begin{equation}
	\begin{bmatrix}
	A-\lambda I \\ C \\ F
	\end{bmatrix}
	=
	\begin{bmatrix}
	A_{11}-\lambda I & \it 0 \\
	A_{21}		     & A_{22}-\lambda I \\
	C_1				 & \it 0 \\
	F_1				 & F_2
	\end{bmatrix}
	.
	\label{eq.matrixpencil}
	\end{equation}
	Now assume that \eqref{eq.linearfuncobsv2} is satisfied. From \eqref{eq.matrixpencil}, we have $\text{row}(F_2) \subseteq \text{row}(A_{22}-\lambda I)$ for all $\lambda \in \text{spec}(A_{22})$, which implies
		\begin{equation}
			\text{row}(F_2) \subseteq \underset{{\lambda \in \text{spec}(A_{22})}}{\bigcap} U_{\lambda}^\perp = \left(
			\underset{\lambda \in \text{spec}(A_{22})}{\bigoplus} U_{\lambda}\right)^\perp = \emptyset,
		\end{equation}
		 where $\oplus$ is the direct sum operator, $U_{\lambda}$ is the left eigenspace of $A_{22}$ corresponding to the eigenvalue $\lambda$, and the second equality comes from the fact that, for an independent numerical realization of the free parameters, $A$ has a complete set of eigenvectors. This shows that $F_2$ contains only all-zero rows, which contradicts the assumption that $\mathcal T\cap\mathcal X_2\neq\emptyset$. Therefore, we have $\rank\begin{bmatrix} A^\transp-\lambda I & C^\transp & F^\transp\end{bmatrix}^\transp>\rank\begin{bmatrix} A^\transp-\lambda I & C^\transp\end{bmatrix}^\transp$.
%
\end{proof}


\begin{lemma}[Necessity of condition 2]
	$\rank[A^\transp\,\,\, C^\transp\,\,\, F^\transp]>\rank[A^\transp\,\,\, C^\transp]^\transp$ if $\mathcal T\cap\mathcal D\neq \emptyset$.
	\label{lemma.lambda2}
\end{lemma}

\begin{proof} 
	Assume that $\mathcal T\cap\mathcal D\neq\emptyset$. Pick $x_i \in \mathcal{T} \cap\mathcal D$ and there will be a minimal dilation set $\mathcal{D}'\subseteq \mathcal{D}$ that contains $x_i$. Let $\mathcal X_2=\mathcal{D}'$ and $\mathcal X_1 = \mathcal X\backslash \mathcal X_2$, where $|\mathcal X_1|=n-k$ and $|\mathcal X_2|=k$. After applying a permutation of coordinates such that the nodes in $\mathcal X_1$ appear first, the matrices $(A,C,F)$ 
	take the form
	\begin{equation}
	A = \begin{bmatrix} A_{1} & A_{2}\end{bmatrix}, \
	C = \begin{bmatrix} C_1 & C_2 \end{bmatrix}, \
	F = \begin{bmatrix} F_1 & F_2 \end{bmatrix},
	\label{eq.dilationnewform}
	\end{equation}
	\noindent
	where $A_{1}\in\R^{k\times n}$, $A_{2}\in\R^{n\times k}$, $C_1\in\R^{q\times (n-k)}$, $C_2\in\R^{q\times k}$, $F_1\in\R^{r\times (n-k)}$, and $F_2\in\R^{r\times k}$. Since $\mathcal{D}'$ is a dilation set, $[A_2^\transp \,\,\, C_2^\transp]^\transp$ has at most $k-1$ nonzero rows due to Remark \ref{remark.dilationrankdeff}. In addition, $F_2$ contains at least one nonzero row because $x_i \in \mathcal{T}$. Let us pick $\bm\phi^\transp \in \mathbb{R}^{1\times k}$, which is the nonzero row of $F_2$ that corresponds to the target node $x_i$. By Assumption \ref{assump.independentF}, the row vector $\bm\phi^\transp$ contains a single nonzero entry in the corresponding column of $x_i$.	
	Now we are ready to prove the Lemma by contradiction. Assuming that $\rank[A^\transp\,\,\, C^\transp\,\,\, F^\transp]^\transp=\rank[A^\transp\,\,\, C^\transp]^\transp$, it follows that
	$\bm\phi^\transp \in \text{row} ([A_2^\transp \,\,\, C_2^\transp]^\transp)$, i.e., there exists a nonzero vector $\bm y^\transp \in \mathbb{R}^{1\times (n+q)} $ such that $\bm\phi^\transp = \bm y^\transp[A_2^\transp \,\,\, C_2^\transp]^\transp$. This implies $\bm y^\transp[A_2^\transp \,\,\, C_2^\transp]_{\mathcal{D}'\setminus \{x_i\}}^\transp = \it 0_{1\times (k-1)}$, where $[A_2^\transp \,\,\, C_2^\transp]_{\mathcal{D}'\setminus \{x_i\}}^\transp$ is the matrix $[A_2^\transp \,\,\, C_2^\transp]^\transp$ after removing the corresponding column of state variable $x_i$. It follows that $\text{rank}([A_2^\transp \,\,\, C_2^\transp]_{\mathcal{D}'\setminus \{x_i\}}^\transp)<k-1$, i.e., $\mathcal{D}'\setminus \{x_i\}$ is also a dilation set. This contradicts the assumption that $\mathcal{D}'$ is a minimal dilation set. Therefore, we conclude that $\rank[A^\transp\,\,\, C^\transp\,\,\, F^\transp]>\rank[A^\transp\,\,\, C^\transp]^\transp$.
\end{proof}

\begin{lemma}
	\cite{VanderWoude1991,VanderWoude1999} 
	Over all choices of free parameters in $(A,B,C,D)$ and $\lambda \in \mathbb{C}$,
	the structural rank of a matrix pencil
	\begin{equation}
	P(\lambda) = \begin{bmatrix}
	A - \lambda I & B \\ C & {D}
	\end{bmatrix}
	\end{equation}
	is equal to $n+l$, where $l$ is the largest number of disjoint paths from the input nodes $u_i\in\mathcal U$ to the output nodes $y_i\in\mathcal S$ in $\mathcal G(A,B,C)$.
	\label{lemma.matrixpencil}
\end{lemma}

\begin{proof}
	See \cite{VanderWoude1991,VanderWoude1999}.
\end{proof}

\begin{lemma}[Sufficiency of condition 1 for \eqref{eq.linearfuncobsv2} for all $\lambda \neq 0$]
	If every state variable $x_i\in\mathcal T$ has a path to some output variable $y_i\in\mathcal S$, then, for almost any independent choice of free parameters in $(A,C,F)$, $\rank\begin{bmatrix}A^\transp-\lambda I & C^\transp & F^\transp\end{bmatrix}^\transp = \rank\begin{bmatrix} A^\transp-\lambda I & C^\transp\end{bmatrix}^\transp$ for every $\lambda\in \mathbb{C}\setminus  \{0\}$.
	\label{lemma.lambda4}
\end{lemma}

\begin{proof}
Let matrix $\bar{P}_i(\lambda)$ be formed by removing the $i$-th row of $\begin{bmatrix}A^\transp-\lambda I & C^\transp & F^\transp\end{bmatrix}^\transp $ and permuting the $i$-th column to the last column, i.e.,
\begin{equation}
\bar{P}_i(\lambda)=	\begin{bmatrix}
	A_i-\lambda I_{n-1} & \bm b_i \\
	C_i & \bm c_i \\
	F_i & \bm f_i
	\end{bmatrix},
\end{equation}
where $A_i $ is the matrix formed by removing the $i$-th row and $i$-th column of $A$; $C_i$ and $F_i$ are formed by removing the corresponding $i$-th column of matrices $C$ and $F$, respectively; $\bm b_i$, $\bm c_i$, and $\bm f_i$ are the $i$-th column of $A-\lambda I$, $C$, and $F$, respectively. In $\mathcal G(A,C)$, this corresponds to removing all incoming links to the $i$-th state variable $x_i$. Thus, by maintaining all outgoing links from $x_i$, we  can view $x_i$ as an input node corresponding to the input vector $\bm b_i $. We further define ${P}_i(\lambda)$ as the matrix formed by removing the $i$-th row from matrix $\begin{bmatrix}A^\transp-\lambda I & C^\transp \end{bmatrix}^\transp $.

For any node $x_i \in \mathcal{T}\cup \mathcal{S}$, by assumption there is a path from $x_i \in \mathcal{T}$ to some node in $\mathcal{S}$. According to Lemma \ref{lemma.matrixpencil}, this leads to $\rank\bar{P}_i(\lambda)=n-1+1=n$. In addition, the path from $x_i $ to some node in $\mathcal{S}$ is unaffected if we remove from $\bar{P}_i(\lambda)$ all rows corresponding to matrix $F$.  Therefore, for independent choices of the free parameters of $(A,C,F)$ and $\lambda \in \mathbb{C}$,
\begin{equation}
	\rank\bar{P}_i(\lambda) = 
	\rank\begin{bmatrix}
	A_i-\lambda I_{n-1} & \bm b_i \\
	C_i & \bm c_i
	\end{bmatrix}
	=
	\rank{P}_i(\lambda) 
	.
	\label{eq.rankcond_lemma4}
\end{equation}

For any node $x_i \notin \mathcal{T}\cup \mathcal{S}$,  we examine separately the cases in which there is and there is not a path from $x_i$ to $\mathcal{T}$. 
 If there is a path from $x_i$ to some node $x_j$ in $\mathcal{T}$, by assumption there is a path from $x_j$ to some node in $\mathcal{S}$,
so we conclude that there is a path from $x_i$ to some node in $\mathcal{S}$. As a result, equation \eqref{eq.rankcond_lemma4} also holds true for this type of node $x_i$. Assume now that there are no paths from $x_i$ to $\mathcal{T}$, which leads to two further possibilities: either there is a path from $x_i$ to $\mathcal{S}$ or no such path exists. In the former case, $\rank\bar{P}_i(\lambda)=\rank{P}_i(\lambda) =n$, while in the latter case we have $\rank\bar{P}_i(\lambda)=\rank{P}_i(\lambda) =n-1$. As a result, equation \eqref{eq.rankcond_lemma4} still holds true for both cases. In summary, equation \eqref{eq.rankcond_lemma4} holds for all $n$ state variables.

Next, assume that, for some particular choice of $\lambda\in \mathbb{C}$, $\rank\begin{bmatrix}A^\transp-\lambda I & C^\transp & F^\transp\end{bmatrix}^\transp > \rank\begin{bmatrix} A^\transp-\lambda I & C^\transp\end{bmatrix}^\transp$ holds for any independent choice of the free parameters in $(A,C,F)$. Since the rank of both matrices in this inequality 
is upper bounded by $n$, we have $\rank\begin{bmatrix} A^\transp-\lambda I & C^\transp\end{bmatrix}^\transp<n$. 
Thus, with this particular choice of $\lambda$, $\rank{P}_i(\lambda) <n$ for any $x_i \in \mathcal{X}$. 
That is, $\lambda$ is the common root for the polynomial $\xi_i(\lambda) = \text{det}({P}_i(\lambda) ) $, $\forall x_i \in \mathcal{X} $. Meanwhile, a necessary condition for $\rank\begin{bmatrix} A^\transp-\lambda I & C^\transp\end{bmatrix}^\transp<n$ is $\rank( A-\lambda I )<n$ or, equivalently, $\lambda$ is a root of the polynomial $\xi_0(\lambda) = \text{det}(A-\lambda I) $.
Note that each polynomial $\xi_i(\lambda)$ does not depend on the free parameters from the $i$-th row of $A$, and that the polynomial $\xi_0(\lambda)$ does not depend on any free parameters from $C$. Thus, each free parameter in the system matrices $(A,C)$ does not appear in at least one of these polynomials. As a result, any common root of all polynomials must not be a function of any of the free parameters. The only possible common root that does not depend on any numerical realization of the free parameters is $\lambda =0$, which exists only when $\begin{bmatrix} A^\transp & C^\transp\end{bmatrix}^\transp$ is rank deficient. Therefore, for any $\lambda\in \mathbb{C}\setminus  \{0\}$, $\rank\begin{bmatrix}A^\transp-\lambda I & C^\transp & F^\transp\end{bmatrix}^\transp = \rank\begin{bmatrix} A^\transp-\lambda I & C^\transp\end{bmatrix}^\transp$.
\end{proof}

\begin{lemma}[Sufficiency of condition 2 for \eqref{eq.linearfuncobsv2} for $\lambda = 0$]
	$\rank[A^\transp\,\,\, C^\transp\,\,\, F^\transp]=\rank[A^\transp\,\,\, C^\transp]^\transp$ if $\mathcal T\cap\mathcal D=\emptyset$.
	\label{lemma.lambda5}
\end{lemma}

\begin{proof}
Let $\mathcal D$ be the union of all minimal dilation sets of $\mathcal G(A,C)$. Let $\mathcal X_2=\mathcal D$ and $\mathcal X_1 = \mathcal X\backslash \mathcal X_2$, where $|\mathcal X_1|=k$ and $|\mathcal X_2|=n-k$. After applying a permutation of coordinates such that the nodes in $\mathcal X_1$ appear first, the matrices $(A,C,F)$ take the form
\begin{equation}
	A = \begin{bmatrix} A_{1} & A_{2}\end{bmatrix},
	C = \begin{bmatrix} C_1 & C_2 \end{bmatrix},
	F = \begin{bmatrix} F_1 & \it{0} \end{bmatrix},
\end{equation}

\noindent
where $A_{1}\in\R^{k\times n}$, $A_{2}\in\R^{n\times (n-k)}$, $C_1\in\R^{q\times k}$, $C_2\in\R^{q\times (n-k)}$, and $F_1\in\R^{r\times k}$. The second block in $F$ contains only zero entries due to the assumption $\mathcal T\cap\mathcal D=\emptyset$. Now, if we assume that $\text{rank} [A_1^\transp \,\,\, C_1^\transp]^\transp<k$, it follows that there is a subset $\mathcal{D}'\subseteq\mathcal{X}_1$ such that the submatrix formed by the corresponding columns of $A_1$ has less than $|\mathcal{D}'|$ nonzero rows. This means that $\mathcal{X}_1$ contains a dilation and thus contains a minimal dilation set, which contradicts to the assumption that $\mathcal{D}$ is the union of all minimal dilation sets. As a result,  $\text{rank} [A_1^\transp \,\,\, C_1^\transp]^\transp=k$, i.e., $\text{row}(F_1) \subseteq \text{row} ([A_1^\transp \,\,\, C_1^\transp]^\transp)$ and also $\text{row}(F) \subseteq \text{row} ([A^\transp \,\,\, C^\transp]^\transp)$ since the second block of $F$ is all-zeros. Therefore, 	$\rank[A^\transp\,\,\, C^\transp\,\,\, F^\transp]=\rank[A^\transp\,\,\, C^\transp]^\transp$.
\end{proof}

\begin{remark}
	\label{remark.funcobsvgenobsv}
	Note that we can have $\mathcal T = \mathcal X$, $\mathcal T\cap\mathcal D=\emptyset$ if and only if $\mathcal D=\emptyset$, which implies that the inference graph $\mathcal G(A,C)$ has no dilations. Therefore, the conditions of Theorem \ref{theor.structuralfuncobsvdefin} reduce to the conditions of Theorem \ref{theor.structuralobsv}.
\end{remark}

\section{Functional observability and the related literature on target controllability}
\label{sec.relatedworks}

Motivated by the fact that controlling the entire state vector of a dynamical system is often unfeasible in large-scale network applications, Gao \textit{et al.} \cite{Gao2014} proposed the concept of \textit{target controllability}, which is based on the concept of \textit{output controllability} from control theory \cite[Section 9.6]{OgataBook}. Formally, a triple $(A,B,F)$ is said to be target controllable if, for any initial state $\bm z(0)=F\bm x(0)$ and final state $\bm z(t_1)=F\bm x(t_1)$, there exists an input $\bm u(t)$ that steers the target vector $\bm z(t)=F\bm x(t)$ from $\bm z(0)$ to $\bm z(t_1)$ in finite time. A triple $(A,B,F)$ is target controllable if and only if 
\begin{equation}
\rank(F\mathcal C) = r,
\label{eq.outputctrb}
\end{equation}

\noindent
where $\mathcal C = [B \,\,\, AB \,\,\, A^2B \,\,\, \ldots \,\,\, A^{n-1}B]$ is the controllability matrix and $F\in\R^{r\times n}$ determines the target variables to be controlled. Thus, target controllability is a sufficient and necessary condition for the invertibility of $FW_{\rm c}(t)F^\transp$, where $W_{\rm c}(t)=\int_0^t e^{A(t-\tau)}BB^\transp e^{A^\transp(t-\tau)}{\rm d}{\tau}$ is the controllability Gramian, and therefore for the existence of a control law $u(t) = - B^\transp e^{A^\transp(t_1-t)} F^\transp (FW_{\rm c}(t_1)F^\transp)^{-1} F(e^{At_1}\bm x(0)-\bm x(t_1))$ capable of steering from $\bm z(0)$ to $\bm z(t_1)$ in finite time \cite[Section 9.6]{OgataBook}.

Given the duality between controllability and observability, one might naively expect that the dual notion of ``target observability'' could be derived by duality from target controllability. For instance, the dual condition of \eqref{eq.outputctrb} would be
\begin{equation}
\rank(\mathcal O F^\transp) = r,
\label{eq.outputobsv}
\end{equation}

\noindent
where $\mathcal O$ is the observability matrix defined in \eqref{eq.obsvmatrix}. However, while condition \eqref{eq.outputctrb} is sufficient and necessary for the design of a controller capable of driving the state of the target variables to any final state, we argue that condition \eqref{eq.outputobsv} does not lead to the design of an estimator capable of reconstructing the state of the target variables. \arthur{In fact, we now show that functional observability is the appropriate necessary and sufficient condition for the design of such estimator.} 

Formally, a pair $(A,C)$ is (completely) observable if, for any unknown initial state $\bm x(0)$, there exists a finite time $t_1>0$ such that knowledge of the input $\bm u(t)$ and output $\bm y(t)$ over $t\in[0,t_1]$ suffices to uniquely determine $\bm x(0)$. Observability is a sufficient and necessary condition for the invertibility of the observability Gramian $W_{o}(t)=\int_0^t e^{A^\transp\tau}CC^\transp e^{A\tau}{\rm d}{\tau}$, and thus for the reconstruction of the initial condition $\bm x(0)$ from the output measurements $\bm y(t)$ over $t\in[0,t_1]$. That is,
\begin{equation}
    \bm x(0) = W_{\rm o}^{-1}(t_1)\int_0^{t_1}e^{A^\transp\tau}C^\transp y(\tau){\rm d}\tau,
    \label{eq.x0gramian}
\end{equation}

\noindent
where we assume here that $\bm u = 0$ without loss of generality \cite{Chi-TsongChen1999}.

Like target controllability, we show that functional observability is a property that characterizes the sufficient and necessary condition for the unique reconstruction of an unknown initial target state $\bm z(0) = F\bm x(0)$ from $\bm y(t)$ over $t\in[0,t_1]$. Note that \eqref{eq.x0gramian} implies
\begin{equation}
    W_{o}(t_1)\bm x(0) = \int_0^{t_1}e^{A^\transp\tau}C^\transp y(\tau){\rm d}\tau
    \label{eq.obsvgramianx0}
\end{equation}

\noindent
and that there exists some matrix $K$ such that $KW_o(t_1) = F$ if and only if $\operatorname{row}(F)\subseteq\operatorname{row}(W_o)$, where\textemdash because  $\operatorname{row}(W_o) = \operatorname{row}(\mathcal O)$\textemdash the latter holds true  if and only if condition \eqref{eq.linearfuncobsv} is satisfied.
Thus, if the system is functionally observable, multiplying \eqref{eq.obsvgramianx0} by $K$ leads to
\begin{equation}
\begin{aligned}
    F\bm x(0) &= K\int_0^{t_1}e^{A^\transp\tau}C^\transp y(\tau){\rm d}\tau.
\end{aligned}
\end{equation}

\noindent
Therefore, functional observability establishes a sufficient and necessary condition for the unique reconstruction of the initial target state $\bm z(0)=F\bm x(0)$ from $\bm y(t)$.

This proof leads to the conclusion that, even though the target controllability condition \eqref{eq.outputctrb} is not dual to the functional observability condition \eqref{eq.linearfuncobsv} in the form commonly articulated in linear system theory, both properties are related in the sense that they are necessary and sufficient conditions for the existence of a target controller (functional observer) capable of driving (estimating) the state of the desired target variables.
For completeness, we note that the functional observability condition \eqref{eq.linearfuncobsv} implies condition \eqref{eq.outputobsv}.  This follows from the fact that if condition \eqref{eq.linearfuncobsv} is satisfied, then $\operatorname{row}(F)\subseteq\operatorname{row}(\mathcal O)$. The converse, however, is not true as shown in the following counter-example.

\begin{figure}
    \centering
	\includegraphics[width=0.3\columnwidth]{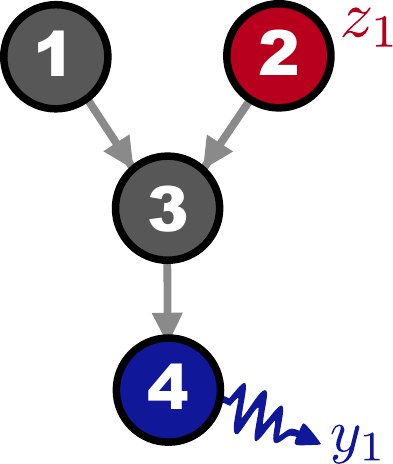}
	\caption{
		Inference graph $\mathcal G(A)$ for the triple $(A,C,F)$ in \eqref{eq.matricestargetobsv}. The state variables $\{x_1,x_2,x_3,x_{4}\}$ are represented by nodes, where  $x_4$ is the sensor node (defined by $\bm y = C\bm x$) and $x_{2}$ is the target node (defined by $\bm z=F\bm x$).
	}
	\label{fig.targetobsv}
\end{figure}

\begin{example}
	\label{examp.outputctrb} 
	
	Consider the complete observability condition \eqref{eq.obsvmatrix}, the functional observability condition \eqref{eq.linearfuncobsv}, and the condition \eqref{eq.outputobsv} for the triple
	\begin{equation}
	A = \begin{bmatrix}
	0 & 0 & 0 & 0 \\ 0 & 0 & 0 & 0 \\ a_{31} & a_{32} & 0 & 0 \\ 0 & 0 & a_{43} & 0
	\end{bmatrix}
	,\,
	C = \begin{bmatrix}
	0 & 0 & 0 & 1
	\end{bmatrix}
	,\,
	F = \begin{bmatrix}
	0 & 1 & 0 & 0
	\end{bmatrix}.
	\label{eq.matricestargetobsv}
	\end{equation}
Fig.~\ref{fig.targetobsv} illustrates the inference graph $\mathcal G(A)$ as well as the target and sensor nodes corresponding to \eqref{eq.matricestargetobsv}. In this scenario, the above conditions reduce to
	\begin{equation}
	\begin{aligned}
	\rank({\mathcal O}) &= 
	\rank\begin{bmatrix}
	0 & 0 & 0 & 1 \\
	0 & 0 & a_{43} & 0 \\
	a_{43}a_{31} & a_{43}a_{32} & 0 & 0 \\
	0 & 0 & 0 & 0 
	\end{bmatrix} = 3 < 4 = n \\ &{\rm (observability)},
	\\
	\rank\begin{bmatrix} \mathcal O \\ F \end{bmatrix} &=  
	\rank\begin{bmatrix}
	0 & 0 & 0 & 1 \\
	0 & 0 & a_{43} & 0 \\
	a_{43}a_{31} & a_{43}a_{32} & 0 & 0 \\
	0 & 0 & 0 & 0 \\
	0 & 1 & 0 & 0
	\end{bmatrix} = 4 > \rank(\mathcal O) \\ &{\rm (functional \,\, observability)},
	\\
	\rank({\mathcal O}F^\transp) &=  
	\rank\begin{bmatrix}
	0 \\ a_{43} \\ 0 \\ 0
	\end{bmatrix} = 1 := r \\ &{\rm (condition \,\, \eqref{eq.outputobsv})}.
	\end{aligned} \label{eq.exampleconditions}
	\end{equation}

	\noindent
	From \eqref{eq.exampleconditions}, we have that the triple $(A,C,F)$ is neither observable nor functionally observable for any choice of parameters in $A$, but it satisfies \eqref{eq.outputobsv}. \QEDA
\end{example}

\section{Minimum-order functional observer design}
\label{sec.graphalgorithm}

\noindent\textbf{Scalability issues of previous numerical procedures.}
Fernando and Trinh \cite{Fernando2010} provided a theoretical solution to the problem of designing a minimum-order functional observer, that is, of finding a minimum-order matrix $F_0$ such that conditions \eqref{eq.darouachcond1}--\eqref{eq.darouachcond2} are satisfied for a triple $(A,C,F_0)$, where $F_0$ is subjected to ${\rm row}(F_0)\supseteq{\rm row(F)}$. For completeness, we show the scalability issues present in the direct numerical implementation \cite{Fernando2010b} of the theoretical results of \cite{Fernando2010} for the design of a minimum-order functional observer. Despite only commenting on this algorithm, we note that other implementations \cite{Fernando2014,Rotella2016,Mohajerpoor2016}
have reported improvements in performance but with no improvement in the scalability of the design algorithm.

A two-stage algorithm is proposed in \cite{Fernando2010b}, where a recursive augmentation of $F_0$ with extra row vectors is carried out in each stage until conditions \eqref{eq.darouachcond1} and \eqref{eq.darouachcond2} are satisfied. The numerical rank conditions in \eqref{eq.darouachcond1} and \eqref{eq.darouachcond2} are computed using SVD, which has complexity of order $O(n^3)$
and thus is not a very scalable approach to be used recursively
(SVD is also often unstable for high-dimensional matrices). 
In a worst-case scenario, $q=r=1$ and one has to estimate $r_0=n-q\approx n$ variables in order to design a stable functional observer. Under these circumstances, the worst-case scenario for the first stage of this algorithm requires finding the minimum $F_0$ that satisfies \eqref{eq.darouachcond1} with $n$ recursive iterations. Since each iteration requires at least one SVD computation, the first stage of the algorithm has complexity of order $O(n^4)$.
Thus, the low scalability of this algorithm is a direct consequence of the use of SVD methods to compute the numerical rank.
For the second stage, the worst-case scenario requires checking the rank condition \eqref{eq.darouachcond2} for up to $\sum_{k=1}^n\binom{n}{k}=2^n-1$ possible submatrices, which has complexity of order $O(2^n)$. Despite this worst-case scenario complexity, usually the second stage of the algorithm requires checking \eqref{eq.darouachcond2} for just a few submatrices, or none at all (for instance, if the conditions of Corollary \ref{corol.structdarouachcond} stated below are satisfied).

\bigskip
\noindent\textbf{Structural conditions for functional observer design.}
With applications in network systems in mind, consider that Assumptions \ref{assump.independentF} and \ref{assump.independentCF} hold for the derivation of the following results.

\begin{corol}
	\label{corol.structdarouachcond}
	If every target node $x_i\in\mathcal T$ has a self-link in $\mathcal G(A,C)$ and condition \eqref{eq.darouachcond1} is true for a structured triple $(A,C,F)$, then condition \eqref{eq.darouachcond2} is true for the same triple.
\end{corol}

\begin{proof}
	Given the condition that every target node $x_i\in\mathcal T$ has a self-link in $\mathcal G(A,C)$, we show that condition \eqref{eq.darouachcond1} implies \eqref{eq.darouachcond2} for all $\lambda$. 
	Consider first condition \eqref{eq.darouachcond2} for $\lambda = 0$. Since, by assumption, condition \eqref{eq.darouachcond1} holds true for $F_0=F$, then condition \eqref{eq.darouachcond2} can be restated for $\lambda = 0$ as
	\begin{equation}
	    \rank\begin{bmatrix}
	    FA \\ C \\ CA
	    \end{bmatrix}
	    =
	    \rank\begin{bmatrix}
	    C \\ CA \\ F \\ FA
	    \end{bmatrix}.
	    \label{eq.corol1proof_lambda0}
	\end{equation}
	
	\noindent
	That is, condition \eqref{eq.corol1proof_lambda0} holds true if $\operatorname{row}(F)\subseteq\operatorname{row}([C^\transp \,\,\, (CA)^\transp \,\,\, (FA)^\transp]^\transp)$. Since condition \eqref{eq.darouachcond1} is true, it follows that $\operatorname{row}(FA)\subseteq\operatorname{row}([C^\transp \,\,\, (CA)^\transp \,\,\, F^\transp]^\transp)$, i.e.,
	\begin{equation}
	    FA = D_1\begin{bmatrix}C \\ CA\end{bmatrix} + D_2 F
	    \label{eq.corol1proof_FAlinearcomb}
	\end{equation}
	
	\noindent
	for some $D_1\in\R^{2q\times n}$ and $D_2\in\R^{r\times r}$. If matrix $D_2$ is invertible, from equation \eqref{eq.corol1proof_FAlinearcomb}, we have that
	\begin{equation}
	    F = D_2^{-1}FA - D_2^{-1}D_1\begin{bmatrix}C \\ CA\end{bmatrix}.
	    \label{eq.Fislineardependent}
	\end{equation}
	
	\noindent
	This means that $\operatorname{row}(F)\subseteq\operatorname{row}([C^\transp \,\,\, (CA)^\transp \,\,\, (FA)^\transp]^\transp)$, satisfying condition \eqref{eq.corol1proof_lambda0}. 
	
	We now show that $D_2$ is indeed invertible if every target node $x_i\in\mathcal T$ has a self-link in $\mathcal G(A,C)$. If this is true, the $i$-th entry of at least one row of $FA$ is a nonzero entry. Moreover, since $x_i$ is a target node, then, from Assumption \ref{assump.independentF}, the $i$-th entry of one row of $F$ is a nonzero entry. As a result, there is always a nonzero value that can be assigned to $[D_2]_{ii}$ such that \eqref{eq.corol1proof_FAlinearcomb} holds true. Because this result holds for all target nodes $x_i\in\mathcal T$, by induction, it follows that $D_2$ has nonzero entries in all of its diagonal elements. Moreover, since $D_2$ is a map between structured matrices $F$ and $FA$,  it is also a structured matrix with independent free parameters \cite{Yamada1985}. 
	As a result, there is always some numerical realization of $(A,C,F)$, and hence of $D_2$, such that $\rank(D_2) = r$ and thus $D_2$ is invertible, which implies \eqref{eq.Fislineardependent} and that condition \eqref{eq.corol1proof_lambda0} is satisfied.
	
	Consider now condition \eqref{eq.darouachcond2} for $\lambda \neq 0$. From the results above, we have that
	\begin{equation}
	    \rank\begin{bmatrix}
	    C \\ CA \\ F
	    \end{bmatrix}
	    =
	    \rank\begin{bmatrix}
	    C \\ CA \\ FA
	    \end{bmatrix}
	    \label{eq.corol1proof_lambda01}.
	\end{equation}
	Since $[(\lambda F-FA)^\transp \,\,\, C^\transp \,\,\, (CA)^\transp]^\transp$ is a linear combination of $[F^\transp \,\,\, C^\transp \,\,\, (CA)^\transp]^\transp$ and $[FA^\transp \,\,\, C^\transp \,\,\, (CA)^\transp]^\transp$, both with the same rank since \eqref{eq.corol1proof_lambda01} is true, it follows that $\rank[(\lambda F-FA)^\transp \,\,\, C^\transp \,\,\, (CA)^\transp]^\transp \leq \rank[F^\transp \,\,\, C^\transp \,\,\, (CA)^\transp]^\transp$. Finally, because $(A,C,F)$ are structured matrices, there is always some numerical realization $(\tilde A,\tilde C,\tilde F)$ such that the equality holds in this relation (and thus condition \eqref{eq.darouachcond2})  for all $\lambda\neq 0$.
\end{proof}

\begin{corol}
	\label{corol.algminF0}
	If $(A,C,F)$ is structurally functionally observable, then Algorithm \ref{alg.findF0} returns a matrix $F_0$ with the smallest order possible such that the rank condition \eqref{eq.darouachcond1} is satisfied for a structured triple $(A,C,F_0)$ subject to the constraint that $F_0$ has only one nonzero entry per row.
\end{corol}

\begin{proof}
	From \cite[Lemma 1]{Fernando2010}, condition \eqref{eq.darouachcond1} can be satisfied for a $F_0$ of minimum order by incrementally augmenting $F_0$ with row vectors from
	\begin{equation}
	\operatorname{row}
	\begin{bmatrix}
	    C \\ CA \\ F_0 \\ F_0A
	\end{bmatrix}
	\cap
	\begin{bmatrix}
	C \\ CA \\ F_0
	\end{bmatrix}
	^\perp.
	\label{eq.rowspaceCCAFFA}
	\end{equation}
	
	\noindent
	Here, we impose the constraint that $F_0$ (and $F$, from Assumption \ref{assump.independentF}) have only one nonzero entry per row. Therefore, condition \eqref{eq.darouachcond1} is only satisfied for a minimum-order $F_0$ if $F_0$ is augmented by the minimum collection of standard basis vectors $\mathcal B$ such that $\operatorname{row}(\mathcal B)$ is a superset of \eqref{eq.rowspaceCCAFFA}, where $\bm b_i\in\mathcal B$ denotes a $n$-dimensional row vector with a nonzero element in the $i$-th coordinate and zeros elsewhere. Since condition \eqref{eq.darouachcond1} is satisfied if and only if $\operatorname{row}(F_0A)\subseteq\operatorname{row}([C^\transp\,\, (CA)^\transp\,\,F_0^\transp]^\transp$, it follows that $\mathcal B$ is the minimum set of standard basis vectors such that $\operatorname{row}(\mathcal B)$ is a superset of \eqref{eq.rowspaceCCAFFA} if
	\begin{equation}
	    \operatorname{row}(\mathcal B) \supseteq \operatorname{row}(F_0A)\cap
	    \begin{bmatrix}
        	C \\ CA \\ F_0
    	\end{bmatrix}
    	^\perp.
	    \label{eq.BinCCAF}
	\end{equation}
	
	Let each element $v'_k\in\mathcal M_2$ correspond to a standard basis vector $\bm b_k$. By the definition of $\mathcal M_2$ in Algorithm \ref{alg.findF0}, $\mathcal B_2 = \{\bm b_k \,:\ v'_k\in\mathcal M_2, \forall k\}$ is the minimum set of standard basis vectors such that $\operatorname{row}(F_0A)\subseteq\operatorname{row}(\mathcal B_2)$. 
	Likewise, the maximum matching search in Algorithm \ref{alg.findF0} determines the set of right-matched nodes $\mathcal M_1$, where each element $v'_j\in\mathcal M_1$ corresponds to a standard basis vector $\bm b_j$. Note that $\mathcal B_1 = \{\bm b_j \, : \, v'_j\in\mathcal M_1,\forall j\}$ is the maximum set of standard basis vectors such that $\operatorname{row}(\mathcal B_1)\subseteq\operatorname{row}([C^\transp \,\, (CA)^\transp \,\, F_0^\transp]^\transp)$. 
	Given $\mathcal B_1$, we can define a minimum set of standard basis vectors $\mathcal B_1'$ such that $\operatorname{row}(\mathcal B_1')=\mathcal B_1^\perp$. By orthogonality, $\mathcal B_1'$ is the minimum set of standard basis vectors such that $([C^\transp \,\, (CA)^\transp \,\, F_0^\transp]^\transp)^\perp\subseteq\mathcal B_1^\perp = \operatorname{row}(\mathcal B_1')$.
    
	Note that the elements of $\mathcal K = \mathcal M_2\backslash\mathcal M_1$ correspond to the standard basis vectors in $\mathcal B=\mathcal B_2\backslash\mathcal B_1 = \mathcal B_2\cap\mathcal B_1'$. Hence, $\operatorname{row}(\mathcal B)= \operatorname{row}(\mathcal B_2)\cap\mathcal B_1^\perp$. Since $\mathcal B_1'$ and $\mathcal B_2$ are the minimum sets of standard basis vectors such that
	$([C^\transp \,\, (CA)^\transp \,\, F_0^\transp]^\transp)^\perp\subseteq\mathcal B_1^\perp$ and $\operatorname{row}(F_0A)\subseteq\operatorname{row}(\mathcal B_2)$, respectively, it can be concluded that $\mathcal B'$ is the minimum set of standard basis vectors such that \eqref{eq.BinCCAF} is satisfied (and, therefore, $\operatorname{row}(\mathcal B)$ is a superset of \eqref{eq.rowspaceCCAFFA}).
	Thus, Algorithm \ref{alg.findF0} returns a minimum-order $F_0$ subject to the stated constraint if $F_0$ is incrementally augmented by standard basis vectors in $\mathcal B$, which are associated with elements in $\mathcal K$.
\end{proof}

Fig.~\ref{fig.algorithm} illustrates the application of Algorithm~\ref{alg.findF0} to find the minimum-order matrix $F_0$ for the triple $(A,C,F)$ in Fig.~\ref{fig.functobsvexample}\textit{A}.

\begin{figure*}
    \centering
	\includegraphics[width=0.6\textwidth]{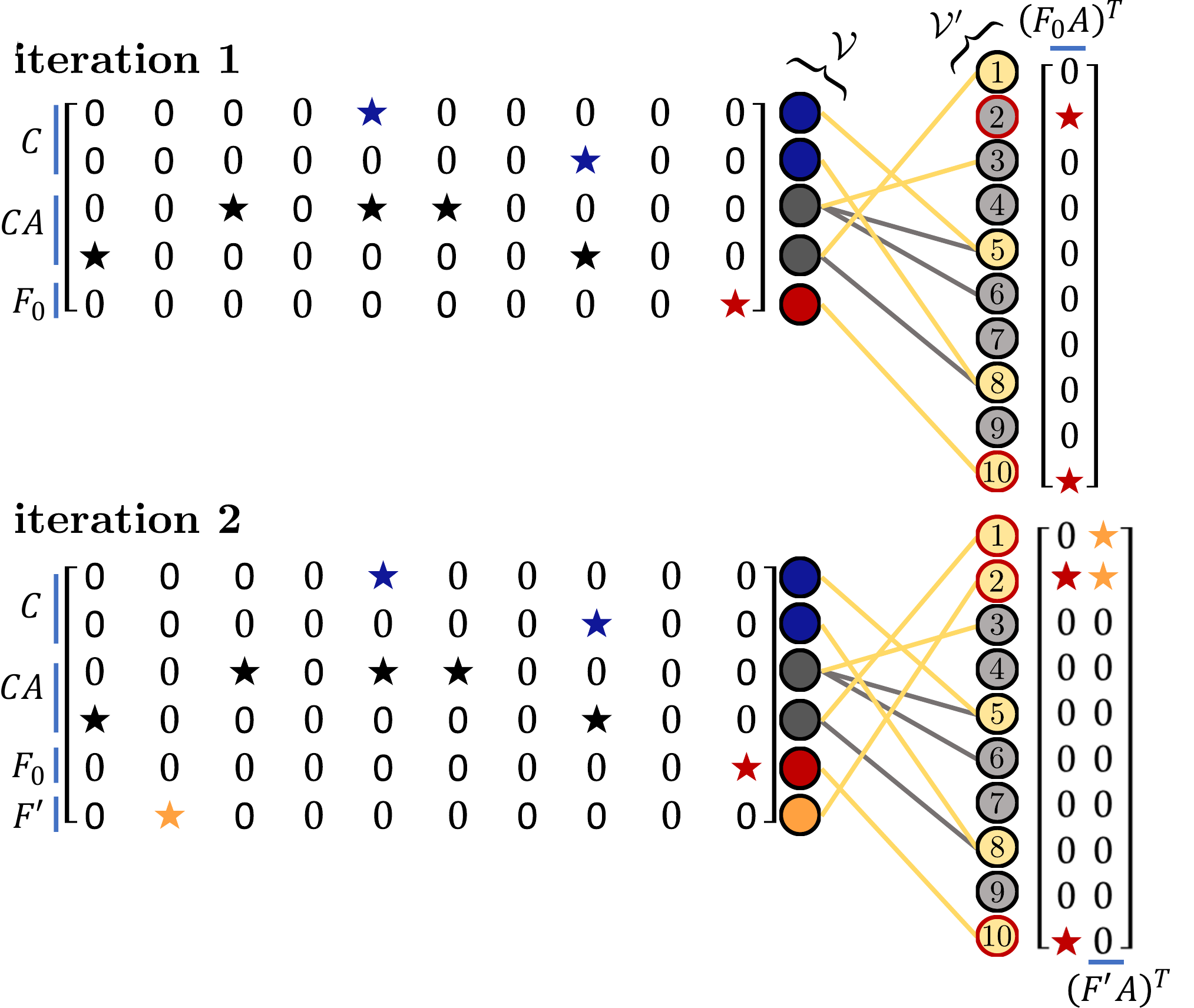}
	\caption{
	Illustrative example of Algorithm~\ref{alg.findF0} applied to the problem in Fig.~\ref{fig.functobsvexample}\textit{A}. On iteration 1, the algorithm builds the corresponding bipartite graph $\mathcal B$ for $[C^\transp \,\,\, (CA)^\transp \,\,\, F_0^\transp]^\transp$, where $F_0\leftarrow F$. Note that a node $v_i\in\mathcal V$, corresponding to a row of $C$ (or $F_0$), is connected to $v_j'\in\mathcal V'$ if $x_j$ is a sensor (or target) node. Likewise, a node $v_i\in\mathcal V$, corresponding to a row of $CA$, is connected to $v_j'\in\mathcal V'$ if $x_j$ is a ``predecessor'' node of a sensor node (i.e., a node with a link pointing to a sensor node). Using a maximum matching algorithm, the matched links $\mathcal E_m$ and right-matched nodes in $\mathcal M_1$ are highlighted in yellow. We highlight the nodes in $\mathcal M_2$ (predecessors of target nodes defined by $F_0$) with a red outline, and define $\mathcal K=\mathcal M_2\backslash\mathcal M_1 = \{v'_2\}$. After picking the element $v'_2$ in $\mathcal K$, and updating $F_0$, the algorithm proceeds to iteration 2. The same steps are repeated. Since $v_1'\in\mathcal M_2$ is already a right-matched node, it follows that $\mathcal K=\emptyset$ and the process terminates. Note that, in a straightforward implementation, the cardinality of $\mathcal V$ increases at every iteration, also increasing the computational burden in the maximum matching computation. We provide a MATLAB implementation of Algorithm~\ref{alg.findF0} that uses an incremental procedure to mitigate this growth and avoid computation of the maximum matching for the whole bipartite graph at every step (\href{https://github.com/montanariarthur/FunctionalObservability}{https://github.com/montanariarthur/FunctionalObservability}). \\
	}
	\label{fig.algorithm}
\end{figure*}

\bigskip
\noindent\textbf{Numerical setup for the comparative analysis of the observers.} 
\arthur{
In Fig.~\ref{fig.obsvestimation}, for individual generated networks defined by a system matrix $A$ 
and some given sets $\mathcal S$ and $\mathcal T$, we have a triple $(A,C,F)$ in which we follow three steps to evaluate the observer performance. 
First, we design a Luenberger observer and a minimum-order functional observer following the procedures described in Methods. 
Second, using a fourth-order Runge-Kutta method, we simulate the dynamics of the system \eqref{eq.linearsys} and observer \eqref{eq.functionalobserverdyn} excited by a step input of $\bm u(t)=10$, $\forall t\geq 0$ (where $B = [1 \ldots 1]^\transp$), with integration step $\delta t = 0.01$, and a total simulation time $t_{\rm f} = 4$ s (which is sufficient to reach a steady-state regime). The initial conditions were randomly drawn from a Gaussian distribution $\mathcal N(\mu,\sigma)$, where $\mu$ is the mean value and $\sigma$ is the standard deviation, as follows: $x_i(0)\sim\mathcal N(0,1)$ and $w_i(0)\sim\mathcal N(10,1)$ for all $i$. 
Third, we compute the estimation error of each observer at each time step $k$ as $\|\bm z(k) - \hat{\bm z}'(k)\|$, where $\bm z(k)=F\bm x(k)\in\R^r$ is the (true vector value) of the target vector and $\hat{\bm z}'(k)\in\R^r$ is the observer estimate of the target vector. Note that $\hat{\bm z}'(k)$ is inferred from the $(n-q)$-dimensional ($r_0$-dimensional) vector $\hat{\bm z}(k)$ estimated in \eqref{eq.functionalobserverdyn} by a Luenberger (functional) observer.
}

\section{Sensor placement and functional observer design in nonlinear systems}
\label{sec.nonlinearfuncobsv}

\arthur{
We show how our graph-theoretic results for the minimum sensor placement and minimum functional observer design can be applied to nonlinear systems. Using model \eqref{eq.sird} as a representative example, this is achieved through the following steps. 
First, we create a ``nonlinear'' inference graph $\bar{\mathcal G}=\{\bar{\mathcal X},\bar{\mathcal E}\}$ of \eqref{eq.sird} (Fig.~\ref{fig.epidemics}\textit{B}), where the bars are used to distinguish from the linear case. Using the notation $\bm x = [S_1 \,\, \ldots \,\, S_N \,\, I_1 \,\, \ldots \,\, I_N \,\, R_1 \,\, \ldots \,\, R_N \,\, D_1 \,\, \ldots \,\, D_N]^\transp$, we identify $x_i\in\bar{\mathcal X}$ as the set of state nodes, $(x_i,x_j)\in\bar{\mathcal E}$ as a directed link from  $x_j$ to $x_i$ if $\dot{x}_i$ is a function of $x_j$ in \eqref{eq.sird}, and $\bar A$ as the resulting binary adjacency matrix of $\bar{\mathcal G}$ \cite{Liu2013c,Aguirre2018}.
Second, we define the set of target nodes $\mathcal T$ and the set of candidate sensor nodes $\mathcal K$, where in our application $I_i\in\mathcal T$ if $i$ is a target city and $\mathcal K=\{D_1,\ldots,D_N\}\backslash\{D_i\,:\,I_i\in\mathcal T\}$. 
We note that there is always a path from $I_i$ to $D_i$ and that $I_i$ always has a self-link in $\bar{\mathcal G}$.
Third, given $(\bar{\mathcal G},\mathcal T,\mathcal K)$, we apply Algorithm~\ref{alg.MSP} to find a reduced set of sensors $\mathcal S\subseteq\mathcal K$ for structural functional observability.
%
%
Fourth, we apply Algorithm~\ref{alg.findF0} to determine matrix $F_0$ (i.e., the set of auxiliary state nodes to be estimated) and, as a result, the functional observer order for $(\bar{A},C,F)$, where $C$ and $F$ are the matrices representing $\mathcal S$ and $\mathcal T$. 
Finally, we apply the results derived in \cite{Trinh2006}, and also reported in \cite[Section 6.2]{Trinh2012}, to compute the functional observer parameters and show that the functional observer satisfies the theoretical conditions for asymptotic convergence of its state estimates, detailed as follows.
}


Consider the following class of nonlinear systems studied in \cite{Trinh2006}:
\begin{equation}
\begin{aligned}
\begin{cases}
\dot{\bm{x}} = A\bm x + \bm f(\bm x), \\
\bm y = C\bm x, \\
\bm z_0 = F_0\bm x,
\end{cases}
\end{aligned}
\label{eq.si.nonlinearclass}
\end{equation}

\noindent
where $\bm f(\bm x):\R^n\mapsto\R^n$ is a nonlinear function not required to be Lipschitz. The epidemiological model \eqref{eq.sird} can be described as \eqref{eq.si.nonlinearclass} by defining $\bm x = [S_1 \,\, \ldots \,\, S_N \,\, I_1 \,\, \ldots \,\, I_N \,\, R_1 \,\, \ldots \,\, R_N \,\, D_1 \,\, \ldots \,\, D_N]^\transp\in\R^n$ and $\bm f(\bm x) = [-\beta_1S_1I_1/P_1 \,\, \ldots \,\, -\beta_NS_NI_N/P_N \,\, \beta_1S_1I_1/P_1 \,\, \ldots \,\, \beta_1S_NI_N/P_N \,\, \mathit{0}_{1\times 2N}]^\transp$, where $n=4N$. Matrix $A\in\R^{n\times n}$ is defined by the linear functions in \eqref{eq.sird}. Matrices $C$ and $F_0$ are determined by Algorithms \ref{alg.MSP} and \ref{alg.findF0}, respectively (third and fourth steps above, with $C$ inferred from $\mathcal S$).

Given $(A,C,F_0)$ and $\bm f(\bm x)$, a reduced-order functional observer can be designed to estimate the partial state vector $\bm z_0$ and thereby $\bm z=F\bm x$ (since $F$ defines the first $r$ rows of $F_0$, see Algorithm \ref{alg.findF0}). Consider the following structure for the functional observer:
\begin{equation}
\begin{aligned}
\begin{cases}
\dot{\bm w} = N\bm w + J\bm y + L\bm f_1(\bm y,\bm z_0), \\
\hat{\bm z} = D\bm w + E\bm y,
\end{cases}
\end{aligned}
\label{eq.si.nonlinearfuncobsv}
\end{equation}

\noindent
where $(N,J,L,D,E)$ and $\bm f_1(\bm x)$ are to be determined such that $\hat{\bm z}(t)$ converges asymptotically to $\bm z_0(t)$. This stable convergence can be achieved by satisfying the conditions stated in \cite[Theorem 6.1]{Trinh2012} as follows.

\begin{enumerate}
	\item For all $i=1,\ldots,n$, decompose $\bm f(\bm x)$ as $\bm f(\bm x) = \bm f_1(\bm y,\bm z_0) + W\bm f_2(\bm x)$, where
	\begin{equation}
	[\bm f_1(\bm z_0)]_i =
	\begin{aligned} 
	\begin{cases}
		-\beta_i \frac{z_j z_k}{P_i}, \,\, \operatorname{for} \,\, i\leq N, \,\, \operatorname{if} \,\, [F_0]_{ji} = 1   \\ \quad\quad \operatorname{and} \,\, [F_0]_{k(i+N)} = 1\,\, \operatorname{for \,\, some} j,k, \\
		+\beta_i \frac{z_j z_k}{P_i}, \,\, \operatorname{for} \,\, N<i\leq 2N, \,\,\,\,  \operatorname{if} \,\, [F_0]_{ji} = 1 \\ \quad\quad  \operatorname{and} \,\, [F_0]_{k(i+N)} = 1\,\, \operatorname{for \,\, some} j,k, \\
		0, \,\,\operatorname{otherwise},
	\end{cases}
	\end{aligned}
	\label{eq.f1decomposed}
	\end{equation} 
	\begin{equation}
	[W\bm f_2(\bm x)]_i =
	\begin{aligned} 
	\begin{cases}
	-\beta_i \frac{x_i x_{(i+N)}}{P_i}, \,\, \operatorname{for} \,\, i\leq N, \,\, \operatorname{if} \,\, [F_0]_{ji} = 0 \,\,  \\ \quad\quad\quad \operatorname{or} \,\, [F_0]_{k(i+N)} = 0 \,\, \operatorname{for \,\, some} j,k, \\
	+\beta_i \frac{x_i x_{(i+N)}}{P_i}, \,\, \operatorname{for} \,\, N<i\leq 2N, \,\, \operatorname{if} \,\, [F_0]_{ji} = 0 \\ \quad\quad\quad \operatorname{or} \,\, [F_0]_{k(i+N)} = 0 \,\, \operatorname{for \,\, some} j,k, \\
	0, \,\,\operatorname{otherwise.}
	\end{cases}
	\end{aligned}
	\label{eq.f2decomposed}
	\end{equation}
	
	Note that $\bm f_1(\bm z_0)$ does not depend on $\bm y$ in the application we are considering, since we defined in the main text that measures are taken only on the number of dead cases $D_i$ (if a given group $i$ is chosen as a ``sensor'' city), and that $\bm f(\bm x)$ is not a function of $D_i$. 
	Importantly, the nonlinear function $\bm f_1(\bm z_0)$, defined in \eqref{eq.f1decomposed}, can be shown to be Lipschitz with a constant Lipschitz constant $\kappa$, i.e.,
	\begin{equation}
		\norm{\bm f_1(\bm z_0) - \bm f_1(\bar{\bm z}_0)} \leq \kappa\norm{\bm z_0 - \bar{\bm z}_0},
		\label{eq.lipschitzf1}
	\end{equation}
	where $\kappa = \max_i(\beta_i)$.
	
	\item Consider the linear transformation in \eqref{eq.lineartransf}, $PW = [W_1^\transp \,\,\, W_2^\transp]^\transp$, and $P^{-1} = [P_1^\transp \,\,\, P_2^\transp]^\transp$. Compute $N_1=(\Phi\Omega^\dagger A_{12}+F_2A_{22})F^\dagger_2$, $N_2=(\Omega\Omega^\dagger - I_q)A_{12}F_2^\dagger$, $\bar L_1 = \Phi\Omega^\dagger P_1 + F_2P_2$ and $\bar L_2 = (\Omega\Omega^\dagger - I_q)P_1$, where $\Omega = [A_{12}F_2^\perp \,\,\, W_1]$ and $\Phi = -[F_2A_{22}F_2^\perp \,\,\, F_2W_2]$. Note that $W$ has a fixed structure (with full-column rank) despite not being a uniquely defined matrix and $\bm f_2(\bm x)$ being treated as an unknown input in the decomposition above.
	Then, it can be verified numerically that there exists matrices $Q=Q^\transp\in\R^{r_0\times r_0}$ and $G\in\R^{r_0\times q}$, and positive scalars $\beta_1$ and $\beta_2$, such that the following linear matrix inequality (LMI) holds:
	\begin{equation}
	\begin{bmatrix}
		\Delta & Q\bar L_1 & G\bar L_2 \\
		\bar L_1^\transp Q & -\beta I_n & \it 0 \\
		\bar L_2^\transp G^\transp & \it 0 & -\beta_2 I_n
	\end{bmatrix}<0,
	\label{eq.si.lmi}
	\end{equation}
	
	\noindent
	where $\Delta = QN_1 + N_1^\transp - GN_2 - N_2^\transp G^\transp + \kappa^2(\beta_1+\beta_2)I_{r_0}$.
\end{enumerate}

\begin{figure}[b!]
    \centering
    \includegraphics[width=\columnwidth]{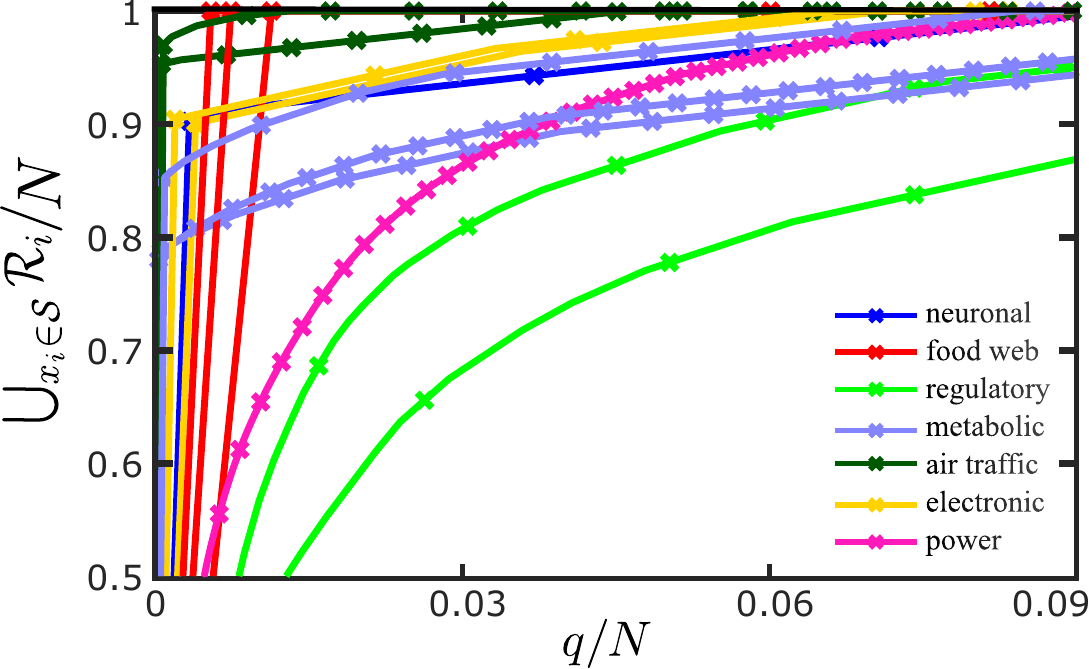}
    \caption{Total number of observable nodes $\cup_{x_i\in \mathcal{S}} \mathcal R_i$ as a function of the number of sensor nodes $q$ (normalized by $N$) in real-world networks. The set $\mathcal S$ is incremented by recursively adding sensors using Algorithm~\ref{alg.MSP} while maximizing the incremental number of observable nodes at each step. 
    See Methods for details on the network models. \\
    }
    \label{fig.si.msp}
\end{figure}

Since there exists a decomposition $\bm f(\bm x) = \bm f_1(\bm y,\bm z_0) + W\bm f_2(\bm x)$ such that $\bm f_1(\bm y,\bm z_0)$ is Lipschitz, and matrices $Q$ and $G$ such that the LMI in \eqref{eq.si.lmi} is satisfied, condition 1 of \cite[Theorem 6.1]{Trinh2012} is satisfied according to \cite[Theorem 6.2]{Trinh2012}. We can now design matrices $(N,J,L,D,E)$ such that condition 2 of \cite[Theorem 6.1]{Trinh2012} is satisfied, and thereby $\hat{\bm z}(t)$ converges asymptotically to $\bm z_0(t)$, as follows \cite[Section 6.2.3]{Trinh2012}: $N = N_1 - ZN_2$, $J = T_1A_{11} + T_2A_{21}-NT_1$, $L = \bar L_1 - Z\bar L_2$, $D = I_{r_0}$, and $E= F_1 - DT_1$, where $Z = Q^{-1}G$, $T_1 = \Phi\Omega^\dagger + Z(I_q-\Omega\Omega^\dagger)$ and $T_2 = F_2$.

\section{Real-world networks}
\label{sec.rwnetsdatabase}

Fig.~\ref{fig.si.msp} shows, for the same real-world networks considered in Fig.~\ref{fig.msp}\textit{B}, the total number of observable nodes as a function of the minimum number of sensors. 
Table \ref{tab.realworldnets} presents the datasets of the real-world networks studied in Figs. \ref{fig.msp}\textit{B} and \ref{fig.si.msp}, which includes networks previously considered in controllability and observability studies \cite{Liu2011,Liu2013c}.

\begin{table*}[b!]
	\caption{Datasets of real-world networks considered in the paper.}
	\label{tab.realworldnets}       
	\centering 
	\begin{tabular}{lllll}
		\toprule[2pt]
		Type & Name & Vertices $N$ & Edges $|\mathcal E|$ & Description\\
		\toprule[2pt]
		Neuronal & \textit{C. elegans} \cite{CelegansNeuralNet,Watts1998}* & $297$ & $2{,}345$ & Neuronal network of \textit{C. elegans}. \\
		\midrule[0.5pt]
		Food web & Grassland \cite{Martinez1999} & $88$ & $137$ & Food web in Grassland. \\
			     & Ythan \cite{Huxham1996} & $135$ & $601$ & Food web in Ythan. \\
				 & Little Rock Lake \cite{Martinez1991}* & $183$ & $2{,}494$ & Food web in Little Rock Lake. \\
		\midrule[0.5pt]
		Regulatory & TRN-Yeast-2 \cite{Milo2002} & $688$ & $1{,}079$ & Transcriptional network of \textit{S. cerevisiae}. \\ 
				 & TRN-EC-2 \cite{Milo2002} & $418$ & $519$ & Transcriptional network of \textit{E. coli}. \\ 
		\midrule[0.5pt]
		Metabolic & \textit{C. elegans} \cite{Overbeek2000}* & $453$ & $2{,}040$ & Metabolic network of \textit{C. elegans}. \\
				  & \textit{E. coli} (iAF1260) \cite{Schell2010} & $1{,}668$ & $6{,}142$ & Metabolic network of \textit{E. coli}. \\
				  & \textit{S. cerevisiae} (iND750) \cite{Schell2010} & $1{,}059$ & $4{,}347$ & Metabolic network of \textit{S. cerevisiae}. \\
				  & \textit{H. sapiens} (RECON1) \cite{Schell2010} & $2{,}766$ & $10{,}280$ & Metabolic network of \textit{H. sapiens}. \\
		\midrule[0.5pt]
		Air traffic & US airports (\url{http://www.transtats.bts.gov})* & $1{,}574$ & $28{,}236$ & Air traffic network between US airports. \\
            		& Air traffic control (\url{http://www.fly.faa.gov})* & $1{,}226$ & $26{,}615$ & Preferred traffic route between US airports. \\
				  
		\midrule[0.5pt]
		Electronics  & s420 \cite{Milo2002}  & $252$ & $399$ & Sequential logic circuit. \\ 
				     & s838 \cite{Milo2002}  & $512$ & $819$ & Sequential logic circuit. \\ 
		\midrule[0.5pt]
		Power grid & Western US power grid \cite{Watts1998}* & $4{,}941$ & $6{,}594$ & Power grid in the western states of the US.\\
		\toprule[2pt]
	\end{tabular}
	\\
	\raggedright
	\footnotesize{\hspace{0.9cm}*These networks are also available through the KONECT database \cite{Kunegis2013}.}
\end{table*}

\end{document}